\documentclass[lettersize,onecolumn]{IEEEtran}
\usepackage{amsmath,amsfonts}
\usepackage{algorithm}
\usepackage{array}
\usepackage[caption=false,font=normalsize,labelfont=sf,textfont=sf]{subfig}
\usepackage{textcomp}
\usepackage{stfloats}
\usepackage{url}
\usepackage{verbatim}
\usepackage{graphicx}
\usepackage{cite}
\usepackage{makecell}

\usepackage{algpseudocode}

\usepackage{tikz}
\usetikzlibrary{arrows.meta,positioning,fit,calc,decorations.pathreplacing,shapes.geometric}

\usepackage{amsthm}
\usepackage{mathrsfs}
\usepackage{bbm}
\usepackage{booktabs}

\newtheorem{definition}{Definition}
\newtheorem{example}{Example}[section]
\newtheorem{theorem}{Theorem}
\newtheorem{remark}{Remark}

\newtheorem{corollary}{Corollary}
\newtheorem{fact}{Fact}
\newtheorem{lemma}{Lemma}

\newcommand{\para}[1]{\emph{#1}\quad}

\usepackage{orcidlink}

\allowdisplaybreaks[4]

\begin{document}

\title{Mutual Information Bounds in the Shuffle Model}

\author{Pengcheng Su\textsuperscript{\orcidlink{0009-0003-9895-4995}}, Haibo Cheng\textsuperscript{\orcidlink{0000-0001-6677-463X}}, Ping Wang\textsuperscript{\orcidlink{0000-0002-8854-2079}},~\IEEEmembership{Senior Member,~IEEE}
\thanks{Pengcheng Su is with the School of Computer Science, Peking University, Beijing, 100091, China. (e-mail: pcs@pku.edu.cn)
Haibo Cheng and Ping Wang are with the National Engineering Research Center of Software Engineering, Peking University, Beiiing,100091,China, and with Key Laboratory of High Confidence Software Technologies, Ministry of Education (Peking University), China. (e-mail: hbcheng@pku.edu.cn, pwang@pku.edu.cn) Haibo Cheng and Ping Wang are the corresponding authors.
}
}

\markboth{}{Su \MakeLowercase{\textit{et al.}}: Mutual Information Bounds in the Shuffle Model}

\maketitle

\begin{abstract}
The shuffle model enhances privacy by anonymizing users’ reports through random permutation. This paper presents the first systematic study of the single-message shuffle model from an information-theoretic perspective. We analyze two regimes: the shuffle-only setting, where each user directly submits its message ($Y_i=X_i$), and the shuffle-DP setting, where each user first applies a local $\varepsilon_0$-differentially private mechanism before shuffling ($Y_i=\mathcal{R}(X_i)$).
Let \(\boldsymbol{Z} = (Y_{\sigma(i)})_i\) denote the shuffled sequence produced by a uniformly random permutation \(\sigma\), and let \(K = \sigma^{-1}(1)\) represent the position of user~1’s message after shuffling. 

For the shuffle-only setting, we focus on a tractable yet expressive \emph{basic configuration}, where the target user’s message follows \(Y_1 \sim P\) and the remaining users’ messages are i.i.d.\ samples from \(Q\), i.e., \(Y_2,\dots,Y_n \sim Q\). 
We derive asymptotic expressions for the mutual information quantities \(I(Y_1;\boldsymbol{Z})\) and \(I(K;\boldsymbol{Z})\) as \(n \to \infty\), and demonstrate how this analytical framework naturally extends to settings with heterogeneous user distributions.

For the shuffle-DP setting, we establish information-theoretic upper bounds on total information leakage. When each user applies an $\varepsilon_0$-DP mechanism, the overall leakage satisfies $I(K; \boldsymbol{Z}) \le 2\varepsilon_0$ and $I(X_1; \boldsymbol{Z}\mid (X_i)_{i=2}^n) \le (e^{\varepsilon_0}-1)/(2n) + O(n^{-3/2})$. These results bridge shuffle differential privacy and mutual-information-based privacy.
\end{abstract}


\section{Introduction}

The shuffle model introduces a \emph{trusted shuffler} that collects messages from users, applies a random permutation, and releases the anonymized multiset to enhance privacy~\cite{Prochlo,cheu19}. 
Such a shuffler can be realized using cryptographic primitives~\cite{Tor,ando_et_al:LIPIcs.ITC.2021.9,10.1145/3319535.3354238,Blinder,Express263836,Atom,10.1145/3576915.3623215}. 
Depending on the number of messages each user submits, the model is categorized as either the \emph{single-message} or the \emph{multi-message} shuffle model.

This work focuses on the \emph{single-message shuffle model}, which has been shown to substantially \emph{amplify} the privacy guarantees of local differential privacy (LDP) mechanisms (see Section~\ref{sec_background_dp})~\cite{Erlingsson19,cheu19,Balle2019,Feldman2021}. 
Moreover, this amplification effect is remarkably \emph{universal} and \emph{robust}: it applies to all LDP mechanisms, depends only on the number of honest users, and remains unaffected by the behavior of corrupted ones~\cite{Cheu22,CheuThesis}. 
In practice, the single-message shuffle model captures essential aspects of real-world anonymization systems such as \emph{anonymous voting} and \emph{honey techniques}~\cite{Juels2013,SecurityAnalysisHoneyword,WangHoneyword22,Su26CSF} (see Section~\ref{sec_intro_honeyword}).

Despite the extensive body of work on differential privacy amplification, the \emph{mutual-information} perspective of the shuffle model has received relatively little attention. 
This paper bridges this gap by providing the first systematic information-theoretic analysis of leakage in the \emph{single-message shuffle model}. 
We study two regimes: the \emph{shuffle-only} setting (\(Y_i = X_i\)) and the \emph{shuffle-DP} setting (\(Y_i = \mathcal{R}(X_i)\)), as illustrated in Fig.~\ref{fig:shuffle_model}.
Prior research has shown that for an \emph{uninformed adversary}—one who knows only the population distribution from which the inputs \(X_i\) are drawn identically and independently—shuffling alone can already provide substantial privacy protection~\cite{10221886}. 
In contrast, for an \emph{informed adversary}\footnote{Also referred to as a \emph{strong adversary} in some prior literature~\cite{Cuff16}.}—who knows all users’ inputs except that of user~1—mere shuffling completely reveals the target input $X_1$, making local randomization indispensable.

For the shuffle-only setting, we analyze a simple yet expressive \emph{basic configuration} where \(Y_1 \sim P\) and \(Y_2, \dots, Y_n \stackrel{\text{i.i.d.}}{\sim} Q\). 
Let \(\boldsymbol{Z} = (Y_{\sigma(i)})_i\) denote the shuffled sequence produced by a uniformly random permutation \(\sigma\), and let \(K = \sigma^{-1}(1)\) represent the position of user~1’s message after shuffling. 
We derive asymptotic expressions for the mutual information quantities \(I(Y_1;\boldsymbol{Z})\) and \(I(K;\boldsymbol{Z})\) as \(n \to \infty\). 
The main theoretical results are summarized in Table~\ref{table1}. 
Throughout this paper, the logarithm in 
\(H(X) = -\sum_i p_i \log(p_i)\) 
is taken to be the \textbf{natural logarithm}, which leads to more concise expressions.  
We further show that our analytical framework extends naturally to heterogeneous settings via the \emph{blanket decomposition} technique~\cite{Balle2019,su2025decompositionbasedoptimalboundsprivacy}, in which the users’ inputs are mutually independent with arbitrary marginals \(X_i \sim P_i\).

For the shuffle-DP setting, we show that the same \emph{blanket decomposition} technique reduces the analysis of a local-DP-then-shuffle mechanism to a homogeneous instance that is analytically tractable. Using this reduction, we derive information-leakage bounds against an informed adversary who knows \(\boldsymbol{X}_{-1}=(X_2,X_3,\dots,X_n)\); no independence assumption on the \(X_i\) is required. For the position leakage \(I(K;\boldsymbol{Z})\), the bound even holds for an adversary who knows the entire input vector \(\boldsymbol{X}=(X_1,\dots,X_n)\).

Let \(\mathcal{P}(\mathbb{X}^n)\) denote the set of all joint distributions on the database \(\boldsymbol{X}=(X_1,\dots,X_n)\), allowing arbitrary dependence.  
For any \(P_{\boldsymbol{X}}\in\mathcal{P}(\mathbb{X}^n)\), let \(I_{P_{\boldsymbol{X}}}(\cdot)\) denote mutual information computed under the joint law induced by \(P_{\boldsymbol{X}}\), the local randomizer \(\mathcal{R}\), and the shuffler.  
Our main bounds are as follows: if \(\mathcal{R}\) satisfies \(\varepsilon_0\)-LDP, then, uniformly over all \(P_{\boldsymbol{X}}\in\mathcal{P}(\mathbb{X}^n)\),
\[
I_{P_{\boldsymbol{X}}}\!\bigl(K;\boldsymbol{Z}\mid \boldsymbol{X}\bigr)\;\le\; 2\varepsilon_0,
\qquad
I_{P_{\boldsymbol{X}}}\!\bigl(X_1;\boldsymbol{Z}\mid \boldsymbol{X}_{-1}\bigr)
\;\le\; \frac{e^{\varepsilon_0}-1}{2n} + O(n^{-3/2}).
\]
In particular, the first inequality implies the unconditional bound \(I_{P_{\boldsymbol{X}}}(K;\boldsymbol{Z}) \le 2\varepsilon_0\) because $K$ is independent of $\boldsymbol{X}$.
These results provide a new information-theoretic characterization of privacy amplification by shuffling.

\begin{figure}[t]
\centering
\resizebox{0.6\linewidth}{!}{%
\begin{tikzpicture}[
  font=\large,
  user/.style={rectangle, draw, rounded corners, minimum width=1.5cm, minimum height=0.6cm},
  rnd/.style={ellipse, draw, minimum width=1.8cm, minimum height=0.8cm},
  shuffler/.style={rectangle, draw, fill=gray!10, minimum width=2.8cm, minimum height=1.2cm},
  msg/.style={rectangle, draw, fill=blue!10, minimum width=1.2cm, minimum height=0.6cm}
]

\node[user] (U1) at (0, 0) {$X_1$};
\node[user] (U2) [below=0.8cm of U1] {$X_2$};
\node (Udots) [below=0.6cm of U2] {$\vdots$};
\node[user] (Un) [below=0.6cm of Udots] {$X_n$};

\node[rnd] (R1) [right=1.8cm of U1] {$\mathrm{Id}\ \text{or}\ \mathcal{R}$};
\node[rnd] (R2) [right=1.8cm of U2] {$\mathrm{Id}\ \text{or}\ \mathcal{R}$};
\node[rnd] (Rn) [right=1.8cm of Un] {$\mathrm{Id}\ \text{or}\ \mathcal{R}$};

\draw[->] (U1) -- (R1);
\draw[->] (U2) -- (R2);
\draw[->] (Un) -- (Rn);

\node[msg] (Y1) [right=1.6cm of R1] {$Y_1$};
\node[msg] (Y2) [right=1.6cm of R2] {$Y_2$};
\node[msg] (Yn) [right=1.6cm of Rn] {$Y_n$};

\draw[->] (R1) -- (Y1);
\draw[->] (R2) -- (Y2);
\draw[->] (Rn) -- (Yn);

\node[align=left] (annot) [above=0.2cm of Y1] {$Y_i=\mathrm{Id}(X_i)\ \text{or}\ Y_i=\mathcal{R}(X_i)$};

\node[shuffler] (S) [right=2.1cm of Y2, yshift=-0.2cm] {Shuffler $\mathcal{S}$};

\draw[->] (Y1) -- ([yshift=0.2cm]S.west);
\draw[->] (Y2) -- (S.west);
\draw[->] (Yn) -- ([yshift=-0.2cm]S.west);

\node[msg] (Z1) [right=2.6cm of S.north east, yshift=0.9cm] {$Z_1 = Y_{\sigma(1)}$};
\node[msg] (Z2) [below=0.5cm of Z1] {$Z_2 = Y_{\sigma(2)}$};
\node (Zdots) [below=0.5cm of Z2] {$\vdots$};
\node[msg] (Zn) [below=0.5cm of Zdots] {$Z_n = Y_{\sigma(n)}$};

\draw[->] (S.east) -- (Z1.west);
\draw[->] (S.east) -- (Z2.west);
\draw[->] (S.east) -- (Zn.west);

\node[align=center] (permnote) [above=0.2cm of S] {$\sigma \sim \mathrm{Uniform}(\mathcal{S}_n)$};

\end{tikzpicture}
}
\caption{Illustration of the single-message shuffle model. Each user~$i$ generates a local message $Y_i$, either deterministically ($Y_i=\mathrm{Id}(X_i)$) or via a local randomizer ($Y_i=\mathcal{R}(X_i)$). The shuffler then applies a random permutation $\sigma$, uniformly sampled from the permutation group $\mathcal{S}_n$, to produce the anonymized outputs $Z_i = Y_{\sigma(i)}$.}
\label{fig:shuffle_model}
\end{figure}
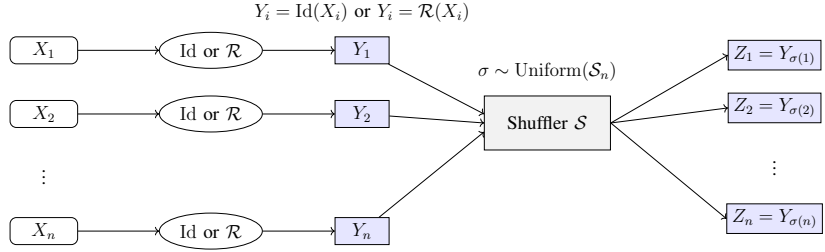

\begin{table}[t]
\centering
\caption{Summary of main results.} 
\label{table1}
\renewcommand{\arraystretch}{2}
\begin{tabular}{l l l}
\toprule
\textbf{Setting} & \textbf{Mutual Information} \(I(K;\boldsymbol{Z})\) & \textbf{Mutual Information} \(I(X_1;\boldsymbol{Z})\) \\ 
\midrule
\textbf{Basic Shuffle-Only (\(P=Q\))} & \makecell[l]{0 \quad (Theorem~\ref{theorem_basic_IK1})} & \makecell[l]{\(\frac{m-1}{2n} + O(n^{-3/2})\) \quad (Theorem~\ref{theorem2})} \\
\textbf{Basic Shuffle-Only (\(P\ll Q\))} & \makecell[l]{\(D_{\mathrm{KL}}(P\Vert Q) - \frac{\chi^2(P\Vert Q)}{2n} + O(n^{-3/2})\) (Theorem~\ref{theorem_IK_general1})} & \makecell[l]{\(\frac{\sum_{i} \frac{p_i - p_i^2}{q_i}}{2n} + O(n^{-3/2})\) (Theorem~\ref{theorem_IY_general1})} \\
\textbf{Basic Shuffle-Only (\(P\not\ll Q\))} & Theorem~\ref{theorem_IK_general2} & Theorem~\ref{theorem_IY_general2} \\
\textbf{General Shuffle-Only (heterogeneous)} & Theorem~\ref{theorem_IK_general_shuffle_only}  & Theorem~\ref{theorem_general_shuffle_IX} \\
\textbf{Shuffle-DP (\(\varepsilon_0\)-LDP then shuffle)} & \makecell[l]{\(I(K;\boldsymbol{Z}\mid \boldsymbol{X})\leq 2\varepsilon_0\) \quad (Theorem~\ref{theorem_shuffle_IK})} & \makecell[l]{\(I(X_1;\boldsymbol{Z}\mid \boldsymbol{X}_{-1})\leq \frac{e^{\varepsilon_0} - 1}{2n} + O(n^{-3/2})\) (Theorem~\ref{theorem_shuffle_IX})} \\
\bottomrule
\end{tabular}

\smallskip
\begin{flushleft}
\footnotesize
Note: \(n\) is the population size; \(m\) denotes the support size of \(P\); 
\(D_{\mathrm{KL}}(P\Vert Q)=\sum_i p_i\log(\frac{p_i}{q_i})\) and \(\chi^2(P\Vert Q)=\sum_i \frac{(p_i-q_i)^2}{q_i}\) are the Kullback–Leibler and chi-squared divergences of \(P\) relative to \(Q\), respectively; 
\(P\ll Q\) means \(\mathrm{Supp}(P) \subseteq \mathrm{Supp}(Q)\). 
The constants hidden in $O(\cdot)$ may depend on $P$, $Q$, or $\mathcal{R}$.
\end{flushleft}
\end{table}

The main contributions of this paper are summarized as follows:
\begin{enumerate}
  \item We provide, to the best of our knowledge, the first information-theoretic characterization of the \emph{single-message shuffle model}. 
  Specifically, we quantify the information leakage from the shuffled output \(\boldsymbol{Z}\) about both the target user’s message position \(K\) and input \(X_1\) via the mutual information quantities \(I(K;\boldsymbol{Z})\) and \(I(X_1;\boldsymbol{Z})\). 
  We analyze both the shuffle-only and shuffle-DP settings, deriving asymptotic bounds on these quantities as functions of the population size \(n\).
  
  \item Our analysis uncovers several structural insights. 
  In the basic shuffle-only setting, \(I(K;\boldsymbol{Z})\) converges asymptotically to the constant \(D_{\mathrm{KL}}(P\Vert Q)\). 
  Interestingly, the distribution \(Q\) that asymptotically minimizes \(I(X_1;\boldsymbol{Z})\) is \emph{not} \(P\). 
  In the shuffle-DP setting, we establish unified upper bounds on the overall mutual-information leakage \(I(X_1;\boldsymbol{Z})\), whose first-order coefficient depends solely on the local privacy parameter~\(\varepsilon_0\).

  \item From a methodological perspective, we extend the \emph{blanket decomposition} technique—originally developed for analyzing privacy amplification in differential privacy—to the information-theoretic domain. 
  This extension highlights the power of decomposition-based approaches in quantifying information leakage within the shuffle model.
\end{enumerate}

The remainder of the paper is organized as follows. 
Section~\ref{sec_background} reviews the necessary background on differential privacy and the shuffle model. 
Section~\ref{sec_shuffle_only} develops our analysis in the shuffle-only setting, and 
Section~\ref{sec_shuffle_dp} extends this framework to the shuffle-DP setting. 
Section~\ref{sec_discussion} discusses an alternative decomposition technique and contrasts it with our approach. 
Section~\ref{sec_conclusion} concludes the paper.

For clarity of presentation, the main body of the paper provides proof sketches and highlights the key ideas, while the detailed technical arguments are deferred to the appendices.

\section{Background}\label{sec_background}

\subsection{Differential Privacy}\label{sec_background_dp}

Differential privacy provides a principled way to quantify the privacy of randomized algorithms. Informally, an algorithm satisfies differential privacy if its output distribution is nearly unchanged when the data of a single individual is modified. In this way, the published outcome reveals very limited information about any particular participant.

The standard (central) model of differential privacy (DP) assumes the existence of a trusted curator that holds the raw data of all $n$ users and applies a DP mechanism on top of it~\cite{Dwork2006}. In many practical deployments, however, the \emph{local} model is preferred: under \emph{local differential privacy} (LDP), each user perturbs their data on the client side before sending it to the collector, thereby removing the need to trust the server~\cite{Cormode18,LDPsurvey,8,Apple17}. The price of this weaker trust assumption, though, is typically a noticeable degradation in utility because the noise must be added per user \cite{6686179,5}.

A line of work therefore considers the \emph{shuffle model}, which inserts a shuffling component between the clients and the analyzer~\cite{cheu19,Erlingsson19,Prochlo}. The shuffler anonymizes the users’ reports by randomly permuting them before they reach the server, and this additional anonymity can be leveraged to improve accuracy. In this paper we concentrate on the \emph{single-message shuffle model}, where each user sends exactly one randomized report and these reports are shuffled prior to aggregation~\cite{feldman2023soda}. Here we give the formal definitions of DP, LDP, and shuffle DP.

We say that random variables \( P \) and \( Q \) are $(\varepsilon, \delta)$-indistinguishable if for all set $T$:
\[
\Pr[P\in T]\le e^{\varepsilon}\Pr[Q\in T]+\delta.
\]
If two datasets \( X \) and \( X' \) have the same size and differ only by the data of a single individual, they are referred to as neighboring datasets (denoted by \( X \simeq X' \)).

\begin{definition}[Differential Privacy]
An algorithm \( \mathcal{R} : \mathbb{X}^n \to \mathbb{Z} \) satisfies $(\varepsilon, \delta)$-differential privacy if for all neighboring datasets \( X, X' \in \mathbb{X}^n \), \( \mathcal{R}(X) \) and \( \mathcal{R}(X') \) are $(\varepsilon, \delta)$-indistinguishable.    

Here, \( \varepsilon \) is referred to as the privacy budget, which controls the privacy loss, while \( \delta \) allows for a small probability of failure. 
When \( \delta = 0 \), the mechanism is also called \( \varepsilon \)-DP.
\end{definition}

\begin{definition}[Local Differential Privacy]
An algorithm \( \mathcal{R} : \mathbb{X} \to \mathbb{Y} \) satisfies local $(\varepsilon, \delta)$-differential privacy if for all \( x, x' \in \mathbb{X} \), \( \mathcal{R}(x) \) and \( \mathcal{R}(x') \) are $(\varepsilon, \delta)$-indistinguishable.    
\end{definition}

A single-message protocol \( \mathcal{P} \) in the shuffle model is defined as a pair of algorithms \( \mathcal{P} = (\mathcal{R}, \mathcal{A}) \), where \( \mathcal{R} : \mathbb{X} \to \mathbb{Y} \), and \( \mathcal{A} : \mathbb{Y}^n \to \mathbb{O} \). We call \( \mathcal{R} \) the \textit{local randomizer}, \( \mathbb{Y} \) the \textit{message space} of the protocol, \( \mathcal{A} \) the \textit{analyzer}, and \( \mathbb{O} \) the \textit{output space}~\cite{Balle2019,cheu19}. 

The overall protocol implements a mechanism \( \mathcal{P} : \mathbb{X}^n \to \mathbb{O} \) as follows: Each user \( i \) holds a data record \( x_i \), to which they apply the local randomizer to obtain a message \( y_i = \mathcal{R}(x_i) \). The messages \( y_i \) are then shuffled and submitted to the analyzer. Let \( \mathcal{S}(y_1, \dots, y_n) \) denote the random shuffling step, where \( \mathcal{S} : \mathbb{Y}^n \to \mathbb{Y}^n \) is a \textit{shuffler} that applies a random permutation to its inputs. 

In summary, the output of \( \mathcal{P}(x_1, \dots, x_n) \) is given by
\[
\mathcal{A} \circ \mathcal{S} \circ \mathcal{R}^n (\boldsymbol{x}) = \mathcal{A}(\mathcal{S}(\mathcal{R}(x_1), \dots, \mathcal{R}(x_n))).
\]

\begin{definition}[Differential Privacy in the Shuffle Model]
A protocol \( \mathcal{P} = (\mathcal{R}, \mathcal{A}) \) satisfies $(\varepsilon, \delta)$-differential privacy in the shuffle model if for all neighboring datasets \( X, X' \in \mathbb{X}^n \), the distributions \( \mathcal{S} \circ \mathcal{R}^n(X) \) and \( \mathcal{S} \circ \mathcal{R}^n(X') \) are $(\varepsilon, \delta)$-indistinguishable.    
\end{definition}

The ``amplification-by-shuffling" theorem in the shuffle model implies that when each of the \(n\) users randomizes their data using an \(\varepsilon_0\)-LDP mechanism, the collection of shuffled reports satisfies \((\varepsilon(\varepsilon_0, \delta, n), \delta)\)-DP, where \(\varepsilon(\varepsilon_0, \delta, n) \ll \varepsilon_0\) for sufficiently large \(n\) and not too small $\delta$~\cite{feldman2023soda}.

Previous research has primarily focused on analyzing and computing the differential privacy guarantees—namely, the parameters $(\varepsilon, \delta)$—achieved by a given protocol $\mathcal{P}$ in the shuffle model \cite{Erlingsson19,Balle2019,cheu19,Feldman2021,feldman2023soda}. Our work is the first to formally study the information leakage in the single-message shuffle model from an information-theoretic perspective.

\begin{remark}
It is worth emphasizing that the functionality of the \emph{multi-message shuffle model} is fundamentally different from that of the single-message setting, and thus the analytical framework developed in this paper does not directly apply. For instance, in the “split-and-mix” protocol, each user sends Shamir additive secret shares of its input to the shuffler, which can be shown to achieve statistically secure addition in the sense of secure multiparty computation (MPC)~\cite{Ishai06}. 

More recent works have further demonstrated that the multi-message shuffle model enables one-round, non-interactive, and information-theoretically secure MPC for arbitrary functions~\cite{cryptoeprint:2025/1442,HIKR23,Hiw25,BEG25}. In this line of research, privacy is formalized within the MPC paradigm rather than through mutual information.
\end{remark}

\subsection{Mutual-Information-Based Privacy}
A fundamental problem in privacy-preserving data analysis is how to \emph{quantify} information leakage (see surveys~\cite{9524532,Wagner18,9380147,10.1145/3604904}). 
Beyond differential privacy (DP)~\cite{Dwork2006}, alternative frameworks have been proposed from the perspectives of \emph{quantitative information flow} (QIF)~\cite{10.3233/JCS-150528,9da52c96268c449d85dacfcd15f27685,10.1109/CSF.2012.26} and \emph{information theory}~\cite{Calmon2012PrivacyAS,6736724,Cuff16}. 

Differential privacy provides a strong, distribution-independent notion of privacy. 
In contrast, QIF quantifies how much the observation of an output $Y$ improves an adversary’s ability to guess the input $X$. 
Let $H_{\infty}(X) = -\log(\max_x p_X(x))$ denote the \emph{min-entropy} of $X$. 
The QIF leakage is defined as $H_{\infty}(X) - H_{\infty}(X\mid Y)$, where 
$H_{\infty}(X\mid Y) = \mathbb{E}_Y[-\log(\max_x p_{X|Y}(x|Y))]$~\cite{10.3233/JCS-150528}. 

From the classical information-theoretic perspective, privacy leakage is measured by the \emph{mutual information} between $X$ and $Y$. 
Let $H(X) = -\sum_x p_X(x)\log p_X(x)$ denote the Shannon entropy. 
Then the information leakage is given by
\begin{align*}
I(X;Y) 
&= H(X) - H(X\mid Y) \\
&= \sum_{x,y} p_{X,Y}(x,y) \log\!\left(\frac{p_{X,Y}(x,y)}{p_X(x)p_Y(y)}\right).
\end{align*}

A substantial line of work has investigated the relationships among these privacy notions. 
Several studies have established formal connections between \emph{central} DP and mutual information~\cite{Cuff16}, or QIF \cite{10.3233/JCS-150528}. Others have examined the relationship between \emph{local} DP and mutual information~\cite{5670946,wang2016mutualinformationoptimallylocal,6686179,10806910}. 
The study in~\cite{10221886} initiated the exploration of the connection between shuffle-DP and QIF-based measures. 
However, to the best of our knowledge, the link between the \emph{shuffle} model of DP and mutual-information-based privacy has not been explicitly characterized prior to this work. 
A related effort~\cite{10.1007/978-3-030-77883-5_16} derived specialized mutual information bounds in the shuffle model to establish lower bounds on estimation error and sample complexity; however, its objectives and analytical techniques differ fundamentally from ours.

\subsection{An Example of the Shuffle-Only Model}\label{sec_intro_honeyword}

Beyond its role in differential privacy, the shuffle model also captures the essence of ``honey'' techniques in computer security. A representative example is the \emph{honeyword} mechanism, a password-leakage detection scheme designed to mitigate the impact of frequent password breaches~\cite{Juels2013}. 

In a honeyword-based system, the authentication server stores a user's true password along with \(n-1\) plausible decoys—called \emph{honeywords}—and then shuffles them so that the genuine password is indistinguishable from the others~\cite{Juels2013}. If an adversary compromises the password database, they obtain a set of \(n\) indistinguishable candidates rather than the true password alone. When a login attempt is made using one of these passwords, the system can detect misuse if a honeyword is entered and raise an alert~\cite{SecurityAnalysisHoneyword,Juels2013}. 

State-of-the-art honeyword generation typically draws decoys i.i.d. from a trained password probability model~\cite{SecurityAnalysisHoneyword,WangHoneyword22}. Let \(P\) represent the distribution of human-chosen passwords, and \(Q\) the distribution used to generate honeywords. This setup precisely corresponds to our basic shuffle-only setting: \(Y_1 \sim P\) represents the genuine password, while \(\boldsymbol{Z}\) denotes the shuffled password list associated with the user.

Existing analyses have primarily focused on the adversary’s success probability after observing the shuffled list, i.e., the probability of identifying the true password \(Y_1\) among the \(n\) entries in \(\boldsymbol{Z}\). More specifically, they study the probability of correctly guessing the \emph{position} (i.e., \(K\) in our notation) of the real password. When \(P = Q\), this success probability is exactly \(1/n\) (details in Section \ref{sec_posterior})~\cite{SecurityAnalysisHoneyword,Su26CSF}. 

However, i.i.d.\ sampling may generate a honeyword that coincides with the real password. In such a case, the adversary can recover the correct password \emph{value} without actually identifying its position. In the extreme case where all \(n\) honeywords are identical, the adversary learns the true password with probability 1.

Our information-theoretic analysis in Section \ref{sec_basic_ll} shows that, when \(P \ll Q\),
\[
  I(Y_1; \boldsymbol{Z}) = \frac{\sum_i \frac{p_i - p_i^2}{q_i}}{n} + O(n^{-3/2}).
\]
A noteworthy consequence is that, asymptotically, the honeyword distribution \(Q\) that minimizes the information about the true password contained in the shuffled list is \emph{not} \(Q = P\); instead, it should satisfy
\(
  q_i \propto \sqrt{p_i(1 - p_i)}.
\)
This result provides new insights into these honey techniques.

\section{Analysis of the Shuffle-Only Setting}\label{sec_shuffle_only}

This section presents our analysis of the shuffle-only setting. 
We begin with a tractable and expressive case, referred to as the \emph{basic shuffle-only setting}, where the target user’s message follows \(Y_1 \sim P\) and all other users’ messages are i.i.d. samples from \(Q\), i.e., \(Y_2, \dots, Y_n \stackrel{\text{i.i.d.}}{\sim} Q\). 
We first derive closed-form expressions for the posterior distributions of the random position \(K = \sigma^{-1}(1)\) (the location of the first user’s message after shuffling) and the corresponding message value \(Y_1 = Z_K\).
Building upon these formulas, we perform exact and asymptotic analyses for three representative cases: \(P = Q\), \(P \ll Q\), and \(P \not\ll Q\).
Finally, we demonstrate that the analytical framework naturally extends to the heterogeneous setting (where users’ inputs follow arbitrary distributions \(X_i \sim P_i\) ) through the \emph{blanket decomposition} technique.

\subsection{Posterior Distributions under the Basic Setting}\label{sec_posterior}

The following theorem characterizes how an observer of the shuffled sequence \(\boldsymbol{Z}\) updates its belief about both the position (\(K\)) of the distinguished user’s message and the message itself (\(Y_1\)).

\begin{theorem}\label{theorem_posterior}
In the shuffle-only setting with basic configuration \((P,Q)\), conditioning on the shuffled output \(\boldsymbol{Z}=\boldsymbol{z}=(z_1,z_2,\dots,z_n)\), the posterior distributions of \(K\) and \(Y_1\) are given by
\begin{align*}
    \Pr[K = k \mid \boldsymbol{Z}=\boldsymbol{z}]
    &=\frac{\tfrac{P(z_k)}{Q(z_k)}}{\sum_{i=1}^n \tfrac{P(z_i)}{Q(z_i)}},\\[4pt]
    \Pr[Y_1=y \mid \boldsymbol{Z}=\boldsymbol{z}]
    &=\sum_{k=1}^n \Pr[K = k \mid \boldsymbol{Z}=\boldsymbol{z}]\,\mathbbm{1}[z_k=y].
\end{align*}
\end{theorem}

The derivation of the posterior distribution of \(K\) follows the approach in prior work~\cite{SecurityAnalysisHoneyword,Su26CSF}; for completeness, the detailed proof is provided in Appendix~\ref{appendix1}.  
The expression for the posterior of the message value \(Y_1 = Z_K\) follows directly from the law of total probability.

\para{Interpretation.}
Theorem~\ref{theorem_posterior} demonstrates that, upon observing the shuffled sequence \(\boldsymbol{Z}\), the posterior probability that the first user’s message occupies position \(k\) is proportional to the likelihood ratio \(\frac{P(z_k)}{Q(z_k)}\). In particular, when \(P = Q\), this ratio equals one for all \(k\), resulting in a uniform posterior distribution that ensures perfect anonymity.

\subsection{Basic Shuffle-Only Setting: $P=Q$}\label{sec_basic_PeqQ}

In this section, we analyze the shuffle information leakage when all users’ messages are drawn from the same distribution $P$, i.e., $P=Q$.  
We focus on two quantities: the position information $I(K;\boldsymbol{Z})$ and the message information\footnote{In the shuffle-only model, we have $I(X_1;\boldsymbol{Z}) = I(Y_1;\boldsymbol{Z})$ since 
$Y_1 = X_1$.  
Throughout the theorems, we therefore use $I(Y_1;\boldsymbol{Z})$, as the corresponding 
analysis extends directly to the study of $I(Y_1;\boldsymbol{Z})$ in the shuffle-DP setting 
(see Section~\ref{sec_discussion}).
} $I(Y_1;\boldsymbol{Z})$.

\begin{theorem}\label{theorem_basic_IK1}
In the basic shuffle-only setting, if $P = Q$, then $I(K;\boldsymbol{Z}) = 0$.
\end{theorem}

\begin{proof}
By Theorem~\ref{theorem_posterior}, the posterior distribution of $K$ conditioned on $\boldsymbol{Z} = \boldsymbol{z}$ is the uniform distribution $U_n$ over $[n]=\{1,2,\dots,n\}$, for every realization $\boldsymbol{z}$. Hence the posterior coincides with the prior, which is also $U_n$. Therefore,
\[
  I(K;\boldsymbol{Z}) = \mathbb{E}_{\boldsymbol{Z}}\bigl[ D_{\mathrm{KL}}( \mathsf{Law}(K \mid \boldsymbol{Z}) \,\Vert\, U_n ) \bigr] = 0,
\]
since each KL term is zero.
\end{proof}

\begin{theorem}\label{theorem2}
In the basic shuffle-only setting, when $P=Q$ and the support size of $P$ is $m$, we have
\[
I(Y_1;\boldsymbol{Z}) \;=\; \frac{m-1}{2n}\;+O\!\left(n^{-3/2}\right),\qquad n\to\infty,
\]
where $\boldsymbol{Z}=(Y_{\sigma(i)})_{i=1}^n$ is the shuffled output of $Y_1,\dots,Y_n\stackrel{iid}{\sim}P$.
\end{theorem}

\begin{proof}
We start by expanding the mutual information as
\begin{align}
I(Y_1;\boldsymbol{Z})
&=H(Y_1)-H(Y_1\mid \boldsymbol{Z}) \notag\\
&=H(Y_1)-\big[H(Y_1\mid \boldsymbol{Z},K)+I(Y_1;K\mid \boldsymbol{Z})\big]. \label{eq:info_expand}
\end{align}
Since conditioning on the shuffle index $K$ makes the message $Y_1$ directly observable in $\boldsymbol{Z}$, we have $H(Y_1\mid \boldsymbol{Z},K)=0$.  
Then, using the chain rule for mutual information,
\[
I(Y_1;K\mid \boldsymbol{Z}) = I(K;Y_1,\boldsymbol{Z}) - I(K;\boldsymbol{Z}).
\]
Substituting this into~\eqref{eq:info_expand} yields
\begin{align}
I(Y_1;\boldsymbol{Z})
&=H(Y_1) - I(K;Y_1,\boldsymbol{Z}) + I(K;\boldsymbol{Z}). \label{eq:info_decomp}
\end{align}
By the previous theorem, $I(K;\boldsymbol{Z})=0$ when $P=Q$.  
Thus, the information leakage reduces to computing $I(K;Y_1,\boldsymbol{Z})$.

\para{Computing $I(K;Y_1,\boldsymbol{Z})$.}
Recall that \(K\) denotes the position in \(\boldsymbol{Z}\) where the true message \(Y_1\) is placed. Conditioned on \(Y_1\) and \(\boldsymbol{Z}\), the posterior distribution of \(K\) is uniform over all indices whose entry in \(\boldsymbol{Z}\) equals \(Y_1\).

Fix \(Y_1 = y_i\). In the shuffled sequence \(\boldsymbol{Z}\), the symbol \(y_i\) appears once from user 1 and, independently, each of the remaining \(n-1\) users contributes \(y_i\) with probability \(p_i = P(y_i)\). Thus the total number of occurrences of \(y_i\) in \(\boldsymbol{Z}\) is distributed as
\[
1 + X, \qquad X \sim \mathrm{Bin}(n-1, p_i).
\]
Given this, the entropy of \(K\) conditioned on \((\boldsymbol{Z}, Y_1 = y_i)\) is the logarithm of that count. Hence
\begin{align}
I(K;Y_1,\boldsymbol{Z})
&=H(K)-H(K\mid \boldsymbol{Z},Y_1) \notag\\
&=\log n - \sum_{i=1}^m p_i\,\mathbb{E}_{X\sim \mathrm{Bin}(n-1,p_i)}\!\big[\log(X+1)\big]. \label{eq:KYZ}
\end{align}
Using the combinatorial identity \(\frac{1}{m} \binom{n-1}{m-1} = \frac{1}{n} \binom{n}{m}\), we can rewrite the term as
\begin{align*}
p_i \, \mathbb{E}_{X \sim \mathrm{Bin}(n-1,p_i)} \big[ \log (X+1) \big]
&= \sum_{j=0}^{n-1} \binom{n-1}{j} p_i^{j+1} (1-p_i)^{n-j-1} \log (j+1) \\
&= \sum_{j=0}^{n-1} \frac{j+1}{n} \binom{n}{j+1} p_i^{j+1} (1-p_i)^{n-j-1} \log (j+1) \\
&= \sum_{j=1}^{n} \frac{j}{n} \binom{n}{j} p_i^{j} (1-p_i)^{n-j} \log j \\
&= \mathbb{E}_{X \sim \mathrm{Bin}(n,p_i)} \Big[ \frac{X}{n} \log X \Big].
\end{align*}
Hence,
\begin{align}
I(K; Y_1, \boldsymbol{Z})
&= \log n - \sum_{i=1}^m \mathbb{E}_{X \sim \mathrm{Bin}(n,p_i)} \Big[ \frac{X}{n} \log X \Big] \notag \\
&= - \sum_{i=1}^m \mathbb{E}_{X \sim \mathrm{Bin}(n,p_i)} \Big[ \frac{X}{n} \log \Big( \frac{X}{n} \Big) \Big]. \label{eq:IKYZ}
\end{align}

\para{Combining terms.}
Substituting~\eqref{eq:IKYZ} into~\eqref{eq:info_decomp} gives
\begin{align}
I(Y_1;\boldsymbol{Z})
&= H(Y_1) + \sum_{i=1}^m \mathbb{E}_{X\sim \mathrm{Bin}(n,p_i)}\!\Big[\frac{X}{n}\log\!\left(\frac{X}{n}\right)\Big] \notag\\
&= \sum_{i=1}^m \mathbb{E}_{X\sim \mathrm{Bin}(n,p_i)}\!\Big[\frac{X}{n}\log\!\left(\frac{X}{n}\right)\Big] - p_i\log p_i. \label{eq:mainI}
\end{align}

\para{Asymptotic expansion.}
Define $f(t)=t\log t$.  
A second-order Taylor expansion of $f(t)$ around $t=p_i$ gives
\begin{equation}\label{eq:taylor}
f(t) = p_i\log p_i + (\log p_i + 1)(t-p_i) + \frac{1}{2p_i}(t-p_i)^2 + O\big((t-p_i)^3\big).
\end{equation}
For $X\sim \mathrm{Bin}(n,p_i)$, the known moments are
\[
\mathbb{E}\!\left[\frac{X}{n}-p_i\right]=0,\quad 
\mathbb{E}\!\left[\Big(\frac{X}{n}-p_i\Big)^2\right]=\mathrm{Var}\!\left(\frac{X}{n}\right)=\frac{p_i(1-p_i)}{n}.
\]
Due to the Rosenthal inequality (Theorem \ref{theorem_Rosenthal_ineq} in Appendix \ref{appendix_2}) for random variables,
$$\mathbb{E}\!\left[\Big|\frac{X}{n}-q_i\Big|^3\right]=O\!\left(n^{-3/2}\right).$$

Taking expectations in~\eqref{eq:taylor} thus yields
\begin{align}
\mathbb{E}_{X\sim \mathrm{Bin}(n,p_i)}\!\Big[\frac{X}{n}\log\!\left(\frac{X}{n}\right)\Big]
&=p_i\log p_i + \frac{1-p_i}{2n} + O\!\left(n^{-3/2}\right). \label{eq:binE}
\end{align}
Some subtle details are provided in Appendix \ref{appendix_2}.

\para{Final result.}
Substituting~\eqref{eq:binE} into~\eqref{eq:mainI} gives
\begin{align}
I(Y_1;\boldsymbol{Z})
&=\sum_{i=1}^m \frac{1-p_i}{2n} + O\!\left(n^{-3/2}\right) \notag\\
&=\frac{m-1}{2n} + O\!\left(n^{-3/2}\right), \label{eq:final}
\end{align}
since $\sum_i (1-p_i)=m-1$.  
This establishes the desired asymptotic expression.
\end{proof}

\begin{example}
We perform numerical experiments under two representative distributions for \(P\):
\begin{itemize}
    \item the uniform distribution on \([m]\);
    \item the Zipf distribution \(\mathrm{Zipf}(m,\alpha)\), where \(p_i \propto i^{-\alpha}\) for \(i=1,2,\dots,m\).
\end{itemize}
The results (Fig. \ref{fig:fig1}) show that the exact value \(I(Y_1;\boldsymbol{Z})\), computed from \eqref{eq:mainI}, matches the asymptotic expression \(\frac{m-1}{2n}\) very closely in both cases.

\begin{figure}[t]
  \centering
  \subfloat[]{%
    \includegraphics[width=0.45\linewidth]{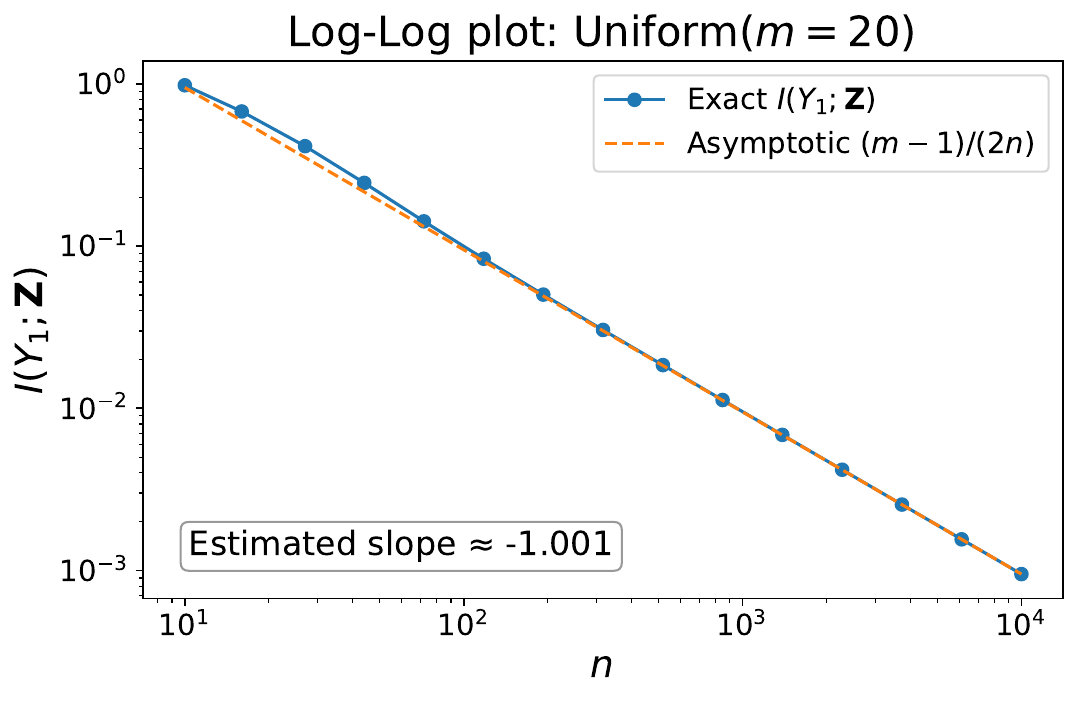}%
    \label{fig1:sub1}
  }
  \hfill
  \subfloat[]{%
    \includegraphics[width=0.45\linewidth]{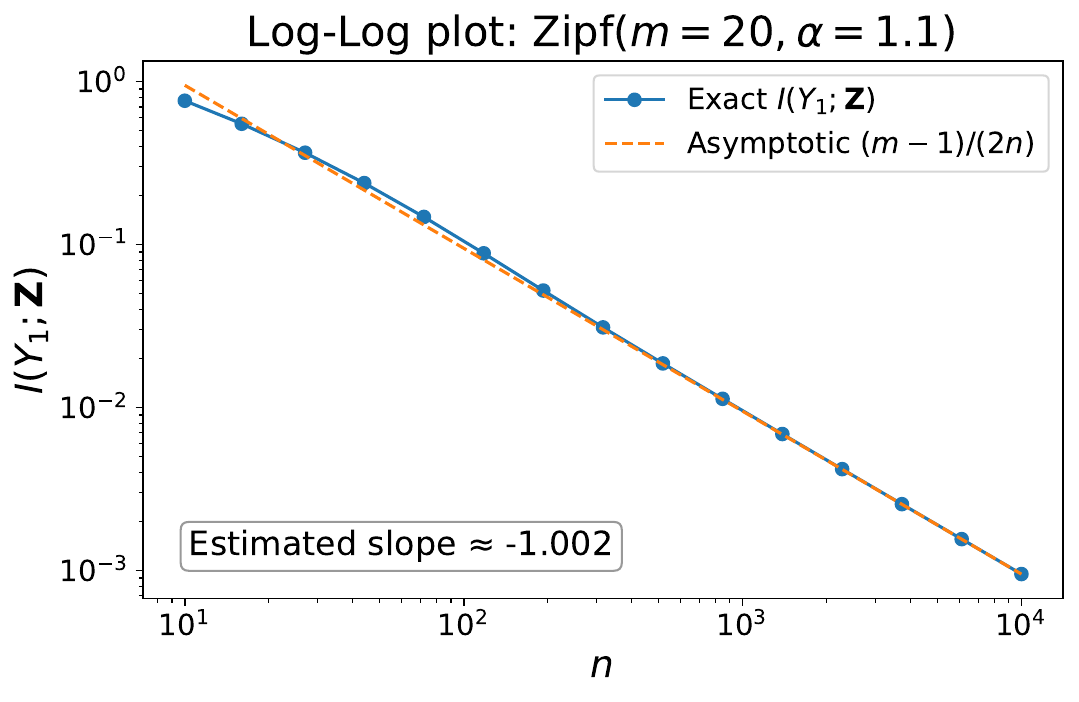}%
    \label{fig1:sub2}
  }
  \caption{Exact vs.\ asymptotic mutual information in the basic shufle-only setting with \(P = Q\).}
  \label{fig:fig1}
\end{figure}
\end{example}

\begin{remark}
When the \(X_i\) are not mutually independent, the leakage \(I(X_1;\boldsymbol{Z})\) can be large even if all marginals are identical. 
For example, under full coherence \(X_1=X_2=\cdots=X_n\) (each with marginal distribution \(P\)), the shuffle-only model outputs \(\boldsymbol{Z}\) as \(n\) copies of the same symbol, thereby revealing \(X_1\) exactly. 
Consequently,
\[
I(X_1;\boldsymbol{Z}) \;=\; H(X_1).
\]    
\end{remark}

\subsection{Basic Shuffle-Only Setting: $P \ll Q$}\label{sec_basic_ll}
Having treated the special case $P = Q$ as a warm-up, we now analyze the basic shuffle-only setting under a general pair of distributions $(P,Q)$. A key distinction here is whether $P$ is absolutely continuous with respect to $Q$ (written $P \ll Q$, i.e., $\mathrm{Supp}(P) \subseteq \mathrm{Supp}(Q)$) or not.

Intuitively, if there exists some symbol $y^{*}$ with $P(y^{*}) > 0$ but $Q(y^{*}) = 0$, then whenever $Y_1 = y^{*}$ the shuffled sequence $\boldsymbol{Z}$ reveals this value completely, leading to nonvanishing leakage. Conversely, when $P \ll Q$, one may expect that the mutual information $I(Y_1; \boldsymbol{Z})$ vanishes as $n \to \infty$.

In what follows, we first handle the “well-aligned’’ case $P \ll Q$, and then extend the analysis to the case $P \not\ll Q$.

\begin{theorem}\label{theorem_IK_general1}
In the basic shuffle-only setting, if \(P \ll Q\), then the mutual information between the shuffle position and the shuffled output satisfies
\[
I(K;\boldsymbol{Z})
\;=\; D_{\mathrm{KL}}(P\Vert Q)\;-\;\frac{\chi^2(P\Vert Q)}{2n}\;+\;O\!\left(n^{-3/2}\right),
\qquad n\to\infty,
\]
where \(\chi^2(P\Vert Q)=\sum_{y} \frac{(P(y)-Q(y))^{2}}{Q(y)}\).
\end{theorem}

\begin{proof}
Recall that \(K\) is uniform on \([n]\) and, conditional on \(\boldsymbol{Z}=\boldsymbol{z}=(z_1,z_2,\dots,z_n)\),
\begin{equation}\label{eq:postK}
\Pr[K=k\mid \boldsymbol{Z}=\boldsymbol{z}]
=\frac{w(z_k)}{\sum_{i=1}^n w(z_i)},\qquad
w(y):=\frac{P(y)}{Q(y)}.
\end{equation}
Hence
\begin{align}
I(K;\boldsymbol{Z})
&= \mathbb{E}_{\boldsymbol{Z}}\!\bigl[ D_{\mathrm{KL}}( \mathsf{Law}(K \mid \boldsymbol{Z}) \,\Vert\, U_n ) \bigr]\notag \\
&= \mathbb{E}_{\boldsymbol{Z}}\!\left[\sum_{k=1}^n \Pr(K=k\mid \boldsymbol{Z}) \log\bigl(n\,\Pr(K=k\mid \boldsymbol{Z})\bigr)\right] \notag\\
&= \log n \;+\; \mathbb{E}_{\boldsymbol{Z}}\!\left[\sum_{k=1}^n \frac{w(Z_k)}{S}\log \frac{w(Z_k)}{S}\right] \notag\\
&= \log n \;+\; \mathbb{E}_{\boldsymbol{Z}}\!\left[\frac{1}{S}\sum_{k=1}^n w(Z_k)\log w(Z_k)\right]
\;-\;\mathbb{E}_{\boldsymbol{Z}}[\log S], \label{eq:I-decomp}
\end{align}
where
\[
S:=\sum_{k=1}^n w(Z_k),\qquad \boldsymbol{Z}=(Z_1,\dots,Z_n).
\]
By construction, among the \(n\) summands in \(S\) there is exactly one draw
\(W^{(P)}:=w(Y_1)\) with \(Y_1\sim P\); the remaining \(n-1\) are i.i.d.\ \(W_i^{(Q)}:=w(Y_i)\) with \(Y_i\sim Q\).
We analyze the two expectations in~\eqref{eq:I-decomp} separately.

\para{The ``ratio'' term \(\mathbb{E}\bigl[\frac{1}{S}\sum w\log w\bigr]\).}
Let
\[
T:=\sum_{k=1}^n w(Z_k)\log w(Z_k)
= W^{(P)}\log W^{(P)} + \sum_{i=2}^n W_i^{(Q)}\log W_i^{(Q)}.
\]
An important property is that $\mathbb{E}_Q W=\sum_y Q(y)\cdot\frac{P(y)}{Q(y)}=1$.

Define
\begin{align*}
 D&:=\mathbb{E}_Q[W\log W]=\sum_{y} Q(y)\,\frac{P(y)}{Q(y)}\log \frac{P(y)}{Q(y)}=D_{\mathrm{KL}}(P\Vert Q),\\
 \chi^2&:=\mathrm{Var}_Q(W)=\sum_{y}Q(y)\,\Bigl(\frac{P(y)}{Q(y)}-1\Bigr)^{2}=\chi^2(P\Vert Q),\\
 B&:=\mathbb{E}_Q[W^{2}\log W].
\end{align*}
Then
\begin{equation}\label{eq:mixed-moments}
\mathbb{E}_P[W]=\mathbb{E}_Q[W^{2}]=\sum_y \frac{P(y)^2}{Q(y)}=1+\chi^2,
\qquad
\mathbb{E}_P[W\log W]=\mathbb{E}_Q[W^{2}\log W]=B.
\end{equation}
With \(\mu_S:=\mathbb{E}[S]=n+\chi^2\) and \(\mu_T:=\mathbb{E}[T]=nD+(B-D)\), a second-order delta-method (Lemma~\ref{lem:ratio-delta} in Appendix \ref{appendix_11}) gives
\begin{equation}\label{eq:ratio-expansion}
\mathbb{E}\!\left[\frac{T}{S}\right]
= \frac{\mu_T}{\mu_S}
\;-\;\frac{\mathrm{Cov}(T,S)}{\mu_S^{2}}
\;+\;\frac{\mu_T\,\mathrm{Var}(S)}{\mu_S^{3}}
\;+\;O(n^{-3/2}).
\end{equation}
Using
\begin{align*}
\mathrm{Var}(S)&=(n-1)\mathrm{Var}_Q(W) + \mathrm{Var}_P(W)= n\chi^2 + O(1),\\
\mathrm{Cov}(T,S)&=(n-1)\,\mathrm{Cov}_Q\!\bigl(W,\,W\log W\bigr)+O(1)\\
&=(n-1)\left( \mathbb{E}_Q[W^2\log W]-\mathbb{E}_Q[W]\cdot \mathbb{E}_Q[W\log W] \right)+O(1)\\
&=(n-1)\,(B-D)+O(1),
\end{align*}
substituting into~\eqref{eq:ratio-expansion} yields
\begin{align}
\mathbb{E}\!\left[\frac{T}{S}\right]
&= \frac{nD+(B-D)}{n+\chi^2}
-\;\frac{\mathrm{Cov}(T,S)}{\mu_S^{2}}
\;+\;\frac{\mu_T\,\mathrm{Var}(S)}{\mu_S^{3}}
\;+\;O(n^{-3/2})\notag \\
&= D \;+\; \frac{1}{n}\Bigl((B-D)-D\chi^2\Bigr)
\;-\;\frac{1}{n}(B-D) \;+\; \frac{1}{n}D\chi^2
\;+\;O(n^{-3/2}) \notag\\
&= D \;+\; O(n^{-3/2}). \label{eq:ratio-final}
\end{align}
Thus the ratio term contributes \(D_{\mathrm{KL}}(P\Vert Q)\) up to \(O(n^{-3/2})\).

\para{The term \(\mathbb{E}[\log S]\).}
Write \(S=n+U\) with \(\mathbb{E}[U]=\chi^2\) and
\(\mathrm{Var}(U)=\mathrm{Var}(S)=n\chi^2+O(1)\).
Expanding \[\log(n+U)=\log n + \frac{U}{n} - \frac{U^{2}}{2n^{2}} + O(\frac{|U|^{3}}{n^{3}}),\] and using
\begin{align*}
\mathbb{E}[U^{2}]&=\mathrm{Var}[U]+(\mathbb{E}[U])^2=n\chi^2+O(1),\\
\mathbb{E}[|U|^3]&=O\Bigl((\mathrm{Var}[U])^{3/2}+(n-1)\mathbb{E}_Q|W-1|^3+\mathbb{E}_P|W-1|^3+|\mathbb{E}[U]|^3\Bigr)=O\!\left(n^{3/2}\right),
\end{align*}
 (by Rosenthal inequality, Theorem \ref{theorem_Rosenthal_ineq}), we obtain
\begin{equation}\label{eq:logS}
\mathbb{E}[\log S]
= \log n + \frac{\chi^2}{n} - \frac{n\chi^2+O(1)}{2n^{2}} + O(n^{-3/2})
= \log n + \frac{\chi^2}{2n} + O(n^{-3/2}).
\end{equation}

\para{Conclusion.}
Combining \eqref{eq:I-decomp}, \eqref{eq:ratio-final}, and \eqref{eq:logS},
\[
I(K;\boldsymbol{Z})
= \log n + D_{\mathrm{KL}}(P\Vert Q)
- \left(\log n + \frac{\chi^2(P\Vert Q)}{2n}\right)
+ O(n^{-3/2}),
\]
which simplifies to the claimed expansion.
\end{proof}

\begin{theorem}\label{theorem_IY_general1}
In the basic shuffle-only setting, when $P\ll Q$, the mutual information between the targeted message and the shuffled output satisfies
\[
I(Y_1;\boldsymbol{Z})
\;=\; \frac{\sum_{i}\frac{p_i-p_i^2}{q_i}}{2n}\;+\;O(n^{-3/2}),
\qquad n\to\infty.
\]
\end{theorem}

\begin{proof}
We analyze the mutual information using the same decomposition as in Eq.~(\ref{eq:info_decomp}):
\[
I(Y_1;\boldsymbol{Z}) = H(Y_1) - I(K;Y_1,\boldsymbol{Z}) + I(K;\boldsymbol{Z}).
\]
Since $I(K;\boldsymbol{Z})$ has already been derived in Theorem \ref{theorem_IK_general1} as
\[
I(K;\boldsymbol{Z}) = D_{\mathrm{KL}}(P\Vert Q) - \frac{\chi^2(P\Vert Q)}{2n} + O(n^{-3/2}),
\]
it remains to compute the term $H(Y_1) - I(K;Y_1,\boldsymbol{Z})$.

\para{Expanding $H(Y_1) - I(K;Y_1,\boldsymbol{Z})$.}
From the derivation of Theorem \ref{theorem2}, we have
\begin{align}
H(Y_1) - I(K;Y_1,\boldsymbol{Z})
&=\sum_{i=1}^m p_i\,\mathbb{E}_{X\sim \mathrm{Bin}(n-1,q_i)}\!\Big[\frac{X+1}{n}\log\!\left(\frac{X+1}{n}\right)\Big]
 - p_i\log p_i. \label{eq:HY1-exp}
\end{align}
Using similar techniques (the combinatorial identity \(\frac{1}{m} \binom{n-1}{m-1} = \frac{1}{n} \binom{n}{m}\)), we have
\begin{align}
H(Y_1) - I(K;Y_1,\boldsymbol{Z})
&=\sum_{i=1}^m \frac{p_i}{q_i}\,
\mathbb{E}_{X\sim \mathrm{Bin}(n,q_i)}\!\Big[\frac{X}{n}\log\!\left(\frac{X}{n}\right)\Big]
 - p_i\log p_i. \label{eq:HY1I}
\end{align}

Let $f(t)=t\log t$.  Expanding around $t=q_i$, we have
\[
f(t) = q_i\log q_i + (\log q_i + 1)(t-q_i) + \frac{1}{2q_i}(t-q_i)^2 + O((t-q_i)^3).
\]
For $X\sim \mathrm{Bin}(n,q_i)$, the moments are
\[
\mathbb{E}\!\left[\frac{X}{n}\right]=q_i,\qquad 
\mathrm{Var}\!\left(\frac{X}{n}\right)=\frac{q_i(1-q_i)}{n}.
\]
Due to the Marcinkiewicz–Zygmund inequality,
$\mathbb{E}\!\left[\Big|\frac{X}{n}-q_i\Big|^3\right]=O\!\left(n^{-3/2} \right)$.

Taking expectations in the expansion gives
\begin{align}
\mathbb{E}_{X\sim \mathrm{Bin}(n,q_i)}\!\Big[\frac{X}{n}\log\!\left(\frac{X}{n}\right)\Big]
&= q_i\log q_i + \frac{1-q_i}{2n} + O(n^{-3/2}). \label{eq:binomialE}
\end{align}

Plugging~\eqref{eq:binomialE} into~\eqref{eq:HY1I} yields
\begin{align}
H(Y_1) - I(K;Y_1,\boldsymbol{Z})
&= \sum_i p_i\log\!\left(\frac{q_i}{p_i}\right)
   + \frac{p_i(1-q_i)}{2q_i n}
   + O(n^{-3/2}). \label{eq:HminusI}
\end{align}
The first summation equals $-D_{\mathrm{KL}}(P\Vert Q)$, and since
\[
\sum_i \frac{p_i(1-q_i)}{q_i}
= -1 + \sum_i \frac{p_i}{q_i},
\]
we can rewrite~\eqref{eq:HminusI} as
\begin{align}
H(Y_1) - I(K;Y_1,\boldsymbol{Z})
= -D_{\mathrm{KL}}(P\Vert Q)
+ \frac{-1+\sum_i \frac{p_i}{q_i}}{2n}
+ O(n^{-3/2}). \label{eq:HminusI2}
\end{align}

Finally, substituting~\eqref{eq:HminusI2} and $I(K;\boldsymbol{Z})$ into the decomposition gives
\begin{align}
I(Y_1;\boldsymbol{Z})
&= H(Y_1) - I(K;Y_1,\boldsymbol{Z}) + I(K;\boldsymbol{Z}) \notag\\
&= \frac{1}{2n}\sum_i \frac{p_i - p_i^2}{q_i}
   + O(n^{-3/2}). \label{eq:IY1Z-final}
\end{align}

This completes the proof.
\end{proof}

\begin{theorem}\label{theorem_optimal_Q}
In the basic shuffle-only setting, for a fixed \(P\), the choice of \(Q\) that asymptotically minimizes the leakage about \(Y_1\) satisfies
\(
q_i \;\propto\; \sqrt{p_i(1-p_i)}.
\)
\end{theorem}

\begin{proof}
According to Theorem \ref{theorem_IY_general1}, the leading term of \(I(Y_1;\boldsymbol{Z})\) is
\(
\frac{1}{2n}\sum_i \frac{p_i(1-p_i)}{q_i}.
\)
Thus, in asymptotic sense, minimizing \(I(Y_1;\boldsymbol{Z})\) over \(Q\) reduces to
\[
\min_{\{q_i\}}\ \sum_i \frac{a_i}{q_i}
\quad \text{subject to}\quad q_i\ge 0,\ \sum_i q_i=1,\quad
\text{where } a_i:=p_i(1-p_i).
\]

By Cauchy--Schwarz inequality,
\[
\sum_i \frac{a_i}{q_i}
=\sum_i \frac{(\sqrt{a_i})^2}{q_i}
\;\ge\;
\frac{\bigl(\sum_i \sqrt{a_i}\bigr)^2}{\sum_i q_i}
=\bigl(\sum_i \sqrt{a_i}\bigr)^2,
\]
with equality if and only if \(\frac{a_i}{q_i}=\lambda q_i\) for some \(\lambda>0\), i.e.,
\(q_i\propto \sqrt{a_i}=\sqrt{p_i(1-p_i)}\). Normalizing to \(\sum_i q_i=1\) yields the stated optimizer.
\end{proof}

\begin{corollary}
If \(\mathrm{Supp}(P)\) has size \(2\), then the asymptotically optimal \(Q\) is the uniform distribution on \(\mathrm{Supp}(P)\).
\end{corollary}

\begin{proof}
Let the two support points have probabilities \(p\) and \(1-p\). Then
\(\sqrt{p(1-p)}=\sqrt{(1-p)p}\) for both points, so the rule
\(q_i\propto \sqrt{p_i(1-p_i)}\) assigns equal mass to each support point.
\end{proof}

\begin{example}
We conduct a numerical experiment with \(P=\mathrm{Zipf}(4,0.7)\) and \(Q\) the uniform distribution on \([4]\). Since closed-form expressions for \(I(K;\boldsymbol{Z})\) and \(I(Y_1;\boldsymbol{Z})\) are intractable, we estimate them via Monte Carlo simulation. In detail,
\[
I(K;\boldsymbol{Z})
= \mathbb{E}_{\boldsymbol{Z}}\!\left[\sum_{k=1}^n \Pr\bigl(K=k \mid \boldsymbol{Z}\bigr)\,
\log\!\bigl(n\,\Pr\bigl(K=k \mid \boldsymbol{Z}\bigr)\bigr) \right],
\]
so we draw many realizations of \(\boldsymbol{Z}\) and approximate the expectation by the empirical average. We use \(100{,}000\) samples to ensure numerical stability. As shown in Fig.~\ref{fig2:sub1}, the estimated mutual informations closely track their asymptotic counterparts.

Next, we empirically verify Theorem~\ref{theorem_optimal_Q} with \(P=\mathrm{Zipf}(4,0.7)\). We compare \(Q_1=P\) with the theoretically optimal choice \(Q_2\) satisfying \(q_i \propto \sqrt{p_i(1-p_i)}\). The corresponding Monte Carlo curves (again with \(100{,}000\) samples) and the theoretical predictions are displayed in Fig.~\ref{fig2:sub2}. The constants are
\(C_1=m-1=3\) and \(C_2=\bigl(\sum_i \sqrt{p_i(1-p_i)}\bigr)^2 \approx 2.81\).
The empirical results align very well with theory.

\begin{figure}[t]
  \centering
  \subfloat[]{%
    \includegraphics[width=0.45\linewidth]{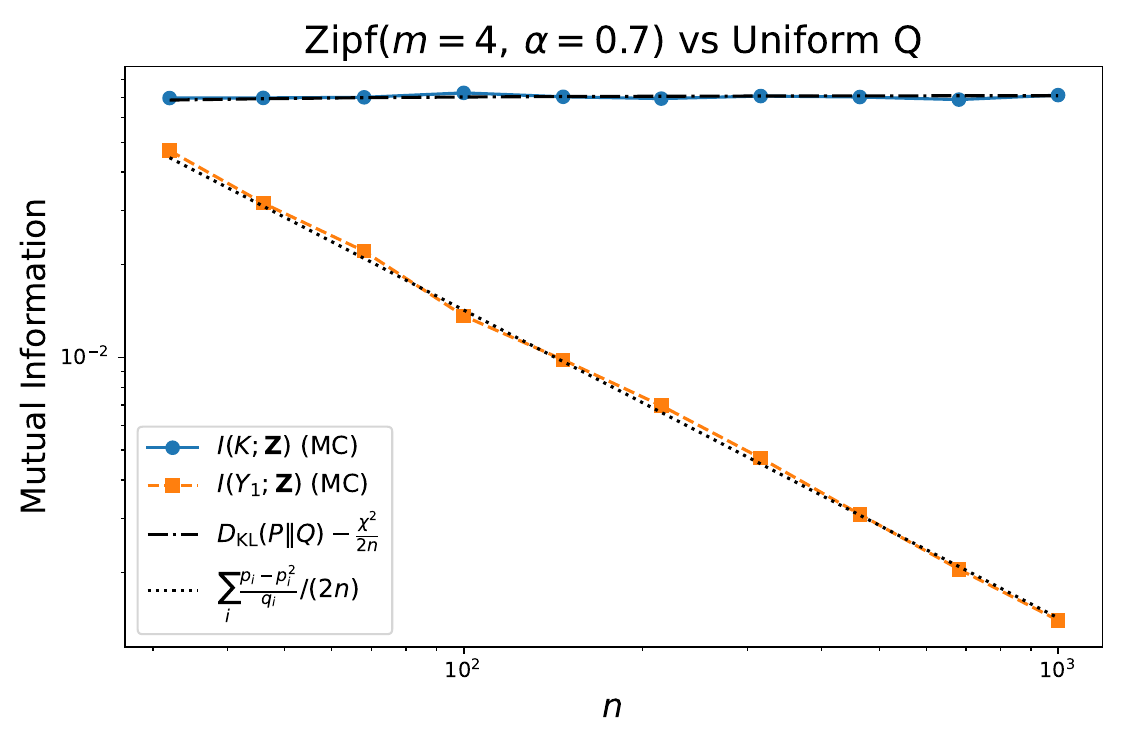}
    \label{fig2:sub1}
  }
  \hfill
  \subfloat[]{%
    \includegraphics[width=0.45\linewidth]{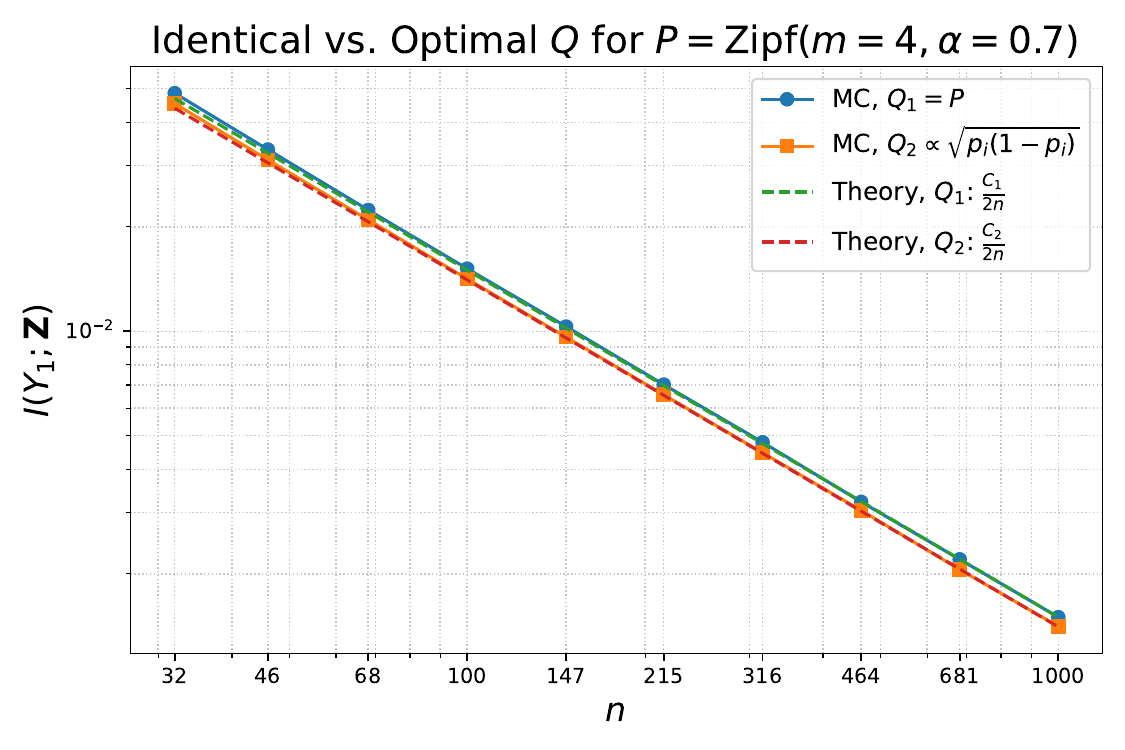}
    \label{fig2:sub2}
  }
  \caption{Exact vs.\ asymptotic mutual information in the basic shuffle-only setting with \(P \ll Q\), and verification of the optimal \(Q\).}
  \label{fig:fig2}
\end{figure}
\end{example}

\subsection{Basic Shuffle-Only Setting: $P \not\ll Q$}

We now extend the analysis from the case $P \ll Q$ to the more general situation where $P$ may assign positive mass to events outside the support of $Q$.
Let $\mathbb{Y} = \mathbb{Y}_1 \cup \mathbb{Y}_2$, where
\[
\mathbb{Y}_2 = \{ y \mid P(y) > 0 \ \text{and}\ Q(y) = 0 \}, 
\qquad
\mathbb{Y}_1 = \mathbb{Y} \setminus \mathbb{Y}_2.
\]
Define the indicator random variable $T = \mathbbm{1}[Y_1 \in \mathbb{Y}_2]$.
When $T=1$, the event $Y_1 \in \mathbb{Y}_2$ implies that $Y_1$ takes a value impossible under $Q$, which makes the shuffled output $\boldsymbol{Z}$ reveal $Y_1$ completely.
Conditioning on $T=0$, the distribution reduces to a restricted case $P' = P \mid_{\mathbb{Y}_1}$ such that $P' \ll Q$.
Let $\beta = \Pr[T=1] = P(\mathbb{Y}_2)$, so that $P = (1-\beta) P\mid_{\mathbb{Y}_1} + \beta P\mid_{\mathbb{Y}_2}$.

The following results generalize the asymptotic expansions obtained in the absolutely continuous case.

\begin{theorem}\label{theorem_IK_general2}
In the basic shuffle-only setting, when $P\not\ll Q$, the mutual information between the shuffle position $K$ and the shuffled output $\boldsymbol{Z}$ satisfies
\[
I(K;\boldsymbol{Z})
\;=\;
\beta \log n
+ (1-\beta)\, D_{\mathrm{KL}}(P'\Vert Q)
- (1-\beta)\, \frac{\chi^2(P'\Vert Q)}{2n}
+ O(n^{-3/2}),
\qquad n\to\infty,
\]
where $\beta=P(\mathbb{Y}_2),\ P'=P|_{\mathbb{Y}\setminus \mathbb{Y}_2},\text{ and } \mathbb{Y}_2 = \{ y \mid P(y) > 0 \ \text{and}\ Q(y) = 0 \}$.
\end{theorem}

\begin{proof}
By the chain rule,
\(
I(K;\boldsymbol{Z},T)
= I(K;T) + I(K;\boldsymbol{Z}\mid T)
= I(K;\boldsymbol{Z}) + I(K;T\mid \boldsymbol{Z}),
\)
we obtain
\[
I(K;\boldsymbol{Z})
= I(K;\boldsymbol{Z}\mid T) + I(K;T) - I(K;T\mid \boldsymbol{Z}).
\]

Observe that $T$ is deterministically recoverable from $\boldsymbol{Z}$ (one simply checks whether
$\boldsymbol{Z}$ contains any element from $\mathbb{Y}_2$).  
Hence $I(K;T\mid \boldsymbol{Z}) = 0$.
Furthermore, $K$ and $T$ are independent:  
$T$ is an intermediate random variable generated in the process
$Y_1\sim P$, whereas $K$ is drawn independently by the shuffler.  
Thus $I(K;T)=0$.

Therefore,
\[
I(K;\boldsymbol{Z})
= I(K;\boldsymbol{Z}\mid T)
= I(K;\boldsymbol{Z}\mid T=0)\Pr[T=0]
 + I(K;\boldsymbol{Z}\mid T=1)\Pr[T=1].
\]

When $T=1$, i.e., $Y_1 \in \mathbb{Y}_2$, the shuffle output deterministically reveals the position of $Y_1$ in $\boldsymbol{Z}$.
Hence $I(K;\boldsymbol{Z}\mid T=1) = \log n$.

When $T=0$, the random variable $Y_1$ follows the restricted distribution $P'\ll Q$. In this case, all previously derived expansions for $I(K;\boldsymbol{Z})$ under $P\ll Q$ (see Theorem~\ref{theorem_IK_general1}) directly apply to $P'$:
\[
I(K;\boldsymbol{Z}\mid T=0)
= D_{\mathrm{KL}}(P'\Vert Q)
- \frac{\chi^2(P'\Vert Q)}{2n}
+ O(n^{-3/2}).
\]
Combining the two parts yields
\[
I(K;\boldsymbol{Z})
=\beta \log n+ (1-\beta)\!\left(D_{\mathrm{KL}}(P'\Vert Q)
- \frac{\chi^2(P'\Vert Q)}{2n}\right)
+ O(n^{-3/2}),
\]
as claimed.
\end{proof}

\begin{theorem}\label{theorem_IY_general2}
In the basic shuffle-only setting, when $P\not\ll Q$, the mutual information between the target message and the shuffled output satisfies
\[
I(Y_1;\boldsymbol{Z})
\;=\;
\sum_{y\in \mathbb{Y}_2}P(y)\log \frac{1}{P(y)}+(1-\beta)\log \frac{1}{1-\beta}+(1-\beta)\frac{\sum_i \frac{p_i' - (p_i')^2}{q_i}}{2n}
+ O(n^{-3/2}),
\qquad n\to\infty,
\]
where $\beta=P(\mathbb{Y}_2),\ P'=P|_{\mathbb{Y}\setminus \mathbb{Y}_2},\text{ and } \mathbb{Y}_2 = \{ y \mid P(y) > 0 \ \text{and}\ Q(y) = 0 \}$.
\end{theorem}

\begin{proof}
Analogously to the proof of Theorem~\ref{theorem_IK_general2}, we obtain
\[
I(Y_1;\boldsymbol{Z})
= I(Y_1;\boldsymbol{Z}\mid T)+I(Y_1;T)
= I(Y_1;\boldsymbol{Z}\mid T=0)\Pr[T=0]
\;+\;
I(Y_1;\boldsymbol{Z}\mid T=1)\Pr[T=1]+I(Y_1;T).
\]

By definition,
\begin{align*}
I(Y_1;T)=H(Y_1)-H(Y_1\mid T)=H(Y_1)-\beta H(Y_1\mid T=1)-(1-\beta)H(Y_1\mid T=0).
\end{align*}
When $T=1$, the shuffle output directly reveals the value of $Y_1$, thus $I(Y_1;\boldsymbol{Z}\mid T=1)
= H(Y_1\mid T=1)$.
When $T=0$, we again invoke Theorem~\ref{theorem_IY_general1} (the case $P'\ll Q$), giving
\[
I(Y_1;\boldsymbol{Z}\mid T=0)
= \frac{\sum_i \frac{p_i' - (p_i')^2}{q_i}}{2n}
+ O(n^{-3/2}).
\]
By definition,
\[
H(Y_1\mid T=0)=\sum_{y\in \mathbb{Y}_1}\frac{P(y)}{1-\beta}\log \frac{1-\beta}{P(y)} .
\]
Combining these parts yields
\begin{align*}
I(Y_1;\boldsymbol{Z})&=H(Y_1)-(1-\beta)H(Y_1\mid T=0)+I(Y_1;\boldsymbol{Z}\mid T=0)\Pr[T=0]\\
&=\sum_{y\in \mathbb{Y}}P(y)\log \frac{1}{P(y)}-\sum_{y\in \mathbb{Y}_1}P(y)\log \frac{1-\beta}{P(y)}+(1-\beta)\frac{\sum_i \frac{p_i' - (p_i')^2}{q_i}}{2n}
+ O(n^{-3/2})\\
&=\sum_{y\in \mathbb{Y}_2}P(y)\log \frac{1}{P(y)}+(1-\beta)\log \frac{1}{1-\beta}+(1-\beta)\frac{\sum_i \frac{p_i' - (p_i')^2}{q_i}}{2n}
+ O(n^{-3/2}).\qedhere
\end{align*}
\end{proof}

\subsection{Analysis of the General Shuffle-Only Setting}\label{sec_general_shuffle_only}

The preceding sections provide a complete analysis of the \emph{basic shuffle-only} setting, in which the other users’ inputs are i.i.d. (\(P_2=P_3=\cdots=P_n=Q\)). In contrast, when the user distributions \(P_i\) are heterogeneous, the posteriors of \(K\) and \(Y_1\) given \(\boldsymbol{Z}\) are no longer available in closed form, unlike the explicit formulas in Section~\ref{sec_posterior}. Fortunately, for analysis of $I(X;\boldsymbol{Z})$, the general shuffle-only setting can still be reduced to the basic one via the \emph{blanket decomposition}.

When applying the blanket decomposition to perform this reduction in the \emph{shuffle-only} model, we assume that the inputs \(X_i\) are \emph{mutually independent}; this ensures that an appropriate post-processing function is well defined. By contrast, in the \emph{shuffle-DP} model no such independence assumption is required, because the post-processing step can be supplied with \(\boldsymbol{X}_{-1}\) as side information.

For the leakage $I(K;\boldsymbol{Z})$, we can obtain bounds under additional conditions on the family $\{P_i\}$.  
The analysis in this case is closely related to the shuffle-DP setting, and we postpone the detailed discussion to Section \ref{sec_shuffle_dp_IK}.

\begin{definition}[Blanket Decomposition~\cite{Balle2019}]
Let \(\gamma := \sum_{y\in\mathbb{Y}} \inf_{i\in[2,n]} P_i(y)\). The \emph{blanket distribution} \(Q^{\mathrm{B}}\) associated with \((P_2,\dots,P_n)\) is
\[
Q^{\mathrm{B}}(y) \;=\; \frac{1}{\gamma}\,\inf_{i\in[2,n]} P_i(y).
\]
Each \(P_i\) then admits the decomposition
\[
P_i \;=\; \gamma\, Q^{\mathrm{B}} \;+\; (1-\gamma)\,\mathrm{LO}_i,
\]
where the \emph{left-over distribution} \(\mathrm{LO}_i\) is given by
\[
\mathrm{LO}_i(y) \;=\; \frac{P_i(y) - \gamma Q^{\mathrm{B}}(y)}{1-\gamma}, \qquad \forall\, y\in\mathbb{Y}.
\]
Equivalently, sampling from \(P_i\) can be implemented by first drawing \(b_i\sim\mathrm{Bern}(\gamma)\); if \(b_i=1\), output a sample from \(Q^{\mathrm{B}}\), otherwise output a sample from \(\mathrm{LO}_i\).
\end{definition}

For technical convenience, we also use the \emph{generalized blanket distribution} \(\bar{Q}^{\mathrm{B}}\)~\cite{su2025decompositionbasedoptimalboundsprivacy}:
\[
\bar{Q}^{\mathrm{B}}(y) \;=\;
\begin{cases}
\inf_{i\in[2,n]} P_i(y), & y\in\mathbb{Y},\\[4pt]
1-\gamma, & y=\perp,
\end{cases}
\]
where \(\perp\) is a placeholder symbol for “non-blanket’’ outcomes, so the output space is \(\bar{\mathbb{Y}}=\mathbb{Y}\cup\{\perp\}\).

\begin{theorem}[Blanket Reduction~\cite{su2025decompositionbasedoptimalboundsprivacy}]
\label{theorem_blanket}
Let \(\boldsymbol{Z}^{r} = \mathcal{S}\bigl(\{Y_1\}\cup\{Y_i' : i=2,\dots,n\}\bigr)\), where \(Y_i' \stackrel{\text{i.i.d.}}{\sim} \bar{Q}^{\mathrm{B}}\) and \(\mathcal{S}\) applies a uniformly random permutation. Then there exists a post-processing function \(f^{\mathrm{B}}\) such that
\[
f^{\mathrm{B}}\!\bigl(\boldsymbol{Z}^{r},\{\mathrm{LO}_i\}_{i=2}^n\bigr) \;\stackrel{d}{=}\; \boldsymbol{Z},
\]
i.e., the distribution of \(\boldsymbol{Z}\) is obtained by post-processing \(\boldsymbol{Z}^{r}\) using \(\{\mathrm{LO}_i\}_{i=2}^n\).
\end{theorem}

Operationally (see Algorithm~\ref{alg:post-processing2}), given \(\boldsymbol{Z}^r\) containing \(c\) occurrences of \(\perp\), the post-processing \(f^{\mathrm{B}}\) selects \(c\) users uniformly at random and replaces those \(\perp\) symbols with independent samples from the corresponding left-over distributions \(\mathrm{LO}_i\). Full proofs and procedural details appear in~\cite{su2025decompositionbasedoptimalboundsprivacy}.

Theorem~\ref{theorem_blanket} shows that a heterogeneous shuffle-only instance \emph{reduces} to a \emph{blanket-mixed} instance in which only the targeted user contributes a genuine message (e.g., \(Y_1\)), while all remaining messages are drawn from \(\bar{Q}^{\mathrm{B}}\).

\begin{theorem}\label{theorem_general_shuffle_IX}
\[
I(X_1;\boldsymbol{Z}) \;\le\; I(X_1;\boldsymbol{Z}^{r}).
\]
\end{theorem}

\begin{proof}
By Theorem~\ref{theorem_blanket}, conditioning on the public parameter set \(\{\mathrm{LO}_i\}_{i=2}^n\), we have the Markov chain
\[
X_1 \;\longrightarrow\; \boldsymbol{Z}^{r} \;\longrightarrow\; \boldsymbol{Z}.
\]
The data-processing inequality then gives \(I(X_1;\boldsymbol{Z}) \le I(X_1;\boldsymbol{Z}^{r})\).
\end{proof}

This reduction is particularly effective when the \(P_i\) are similar. If a particular \(P_i\) deviates substantially from the rest, one may exclude that user from the analysis: the output \(\boldsymbol{Z}\) without that user’s contribution can be post-processed to reconstruct the full \(\boldsymbol{Z}\) by re-inserting that user’s message.

\begin{example}
Suppose \(P_1=(0.4,0.6)\), \(P_2=(0.3,0.7)\), \(P_3=(0.5,0.5)\), and \(P_4=(1.0,0)\). User~4’s distribution differs markedly from the others, so we may exclude it. The generalized blanket distribution of \(P_2\) and \(P_3\) is \(\bar{Q}^{\mathrm{B}}=(0.3,0.5,0.2)\), and
\[
I(X_1;\boldsymbol{Z}) \;\le\; I(X_1;\boldsymbol{Z}^{r}),
\]
where \(\boldsymbol{Z}^{r}=\mathcal{S}\bigl(\{Y_1\}\cup\{Y_i' : i=2,3\}\bigr)\) and \(Y_i' \stackrel{\text{i.i.d.}}{\sim} \bar{Q}^{\mathrm{B}}\).
\end{example}

\section{Analysis of the Shuffle-DP Setting}\label{sec_shuffle_dp}

In this section, we analyze the mutual information under the \emph{shuffle differential privacy} (shuffle-DP) setting.
Each user applies an $\varepsilon_0$-local randomizer $\mathcal{R}$ to their input $X_i$, producing
\(
Y_i = \mathcal{R}(X_i), \ i = 1, \ldots, n,
\)
and the outputs are subsequently shuffled by a uniformly random permutation $\sigma$, yielding the shuffled view
$\boldsymbol{Z} = (Y_{\sigma(i)})_{i=1}^n$. Let $\boldsymbol{X}=(X_1,X_2,X_3,\dots,X_n)$ and $\boldsymbol{X}_{-1}=(X_2,X_3,\dots,X_n)$.
We will demonstrate two results:
\begin{enumerate}
    \item the shuffled output $\boldsymbol{Z}$ is $2\varepsilon_0$-DP with respect to the shuffle index $K$ even given $\boldsymbol{X}$, implying
$I(K;\boldsymbol{Z}\mid \boldsymbol{X}) \le 2\varepsilon_0$;
\item by reduction to the basic setting via the \emph{blanket decomposition}
\cite{Balle2019,su2025decompositionbasedoptimalboundsprivacy},
we have
\(
I(X_1;\boldsymbol{Z}\mid \boldsymbol{X}_{-1})
\;\le\;
\frac{e^{\varepsilon_0}-1}{2n}
+ O(n^{-3/2}).
\)
\end{enumerate}
It is worth noting that these conclusions hold without any independence assumption on the \(X_i\).

\begin{remark}\label{remark3}
The leakage \(I(X_1;\boldsymbol{Z})\) can be large when the \(X_i\) are not mutually independent. 
For instance, under full coherence \(X_1=X_2=\cdots=X_n\) (each with marginal distribution \(P\)), the shuffle-DP model with the randomized response mechanism (Example~\ref{example_rr}) produces \(\boldsymbol{Z}\) as \(n\) i.i.d.\ randomized reports of the \emph{same} underlying symbol. By the law of large numbers, a simple majority-vote decoder recovers \(X_1\) with probability tending to one, so the leakage is almost complete:
\[
I(X_1;\boldsymbol{Z}) \;\longrightarrow\; H(X_1)\qquad\text{as } n\to\infty.
\]   
\end{remark}

For convenience, $\Pr[\mathcal{R}(x)=y]$ is denoted by $\mathcal{R}_x(y)$.

\subsection{Analysis of $I(K;\boldsymbol{Z})$}\label{sec_shuffle_dp_IK}

\begin{theorem}\label{theorem_shuffle_IK}
In the shuffle-DP setting, the shuffled output $\boldsymbol{Z}$ is $2\varepsilon_0$-DP with respect to the shuffle index $K$, for any given $\boldsymbol{X}$.
Consequently,
\[
I(K;\boldsymbol{Z} \mid \boldsymbol{X}) \le 2\varepsilon_0.
\]
\end{theorem}

\begin{proof}
We begin by expressing the posterior of $K$ given $\boldsymbol{Z}=\boldsymbol{z}$ and $\boldsymbol{X}=\boldsymbol{x}$:
\begin{align*}
\Pr[K=k\mid \boldsymbol{Z}=\boldsymbol{z},\boldsymbol{X}=\boldsymbol{x}]
&= \frac{\Pr[K=k\mid \boldsymbol{X}=\boldsymbol{x}]\,\Pr[\boldsymbol{Z}=\boldsymbol{z}\mid K=k,\boldsymbol{X}=\boldsymbol{x}]}{\Pr[\boldsymbol{Z}=\boldsymbol{z},\boldsymbol{X}=\boldsymbol{x}]}.
\end{align*}
Since $K$ is uniformly distributed over $[n]$ and independent of $\boldsymbol{X}$, we have $\Pr[K=k\mid \boldsymbol{X}=\boldsymbol{x}]=1/n$ for all $k$.
Hence, differential privacy with respect to $K$ is determined by bounding the ratio
$\Pr[\boldsymbol{Z}=z\mid K=k,\boldsymbol{X}=\boldsymbol{x}]/\Pr[\boldsymbol{Z}=\boldsymbol{z}\mid K=k',\boldsymbol{X}=\boldsymbol{x}]$ for any two positions $k,k'\in[n]$.

\para{Expansion of the conditional likelihood.}
Conditioned on $K=k$, the message in position $k$ originates from user~1, while all other messages are outputs of the remaining users.
Then
\begin{equation}
\Pr[\boldsymbol{Z}=\boldsymbol{z}\mid K=k,\boldsymbol{X}=\boldsymbol{x}]
=\sum_{\sigma:\sigma(k)=1}\;\frac{1}{(n-1)!}\prod_{i=1}^n \mathcal{R}_{x_{\sigma(i)}}(z_i). \label{eq8}
\end{equation}
Similarly for $K=k'$:
\begin{equation}
\Pr[\boldsymbol{Z}=\boldsymbol{z}\mid K=k',\boldsymbol{X}=\boldsymbol{x}]
=\sum_{\sigma:\sigma(k')=1}\frac{1}{(n-1)!}\;\prod_{i=1}^n \mathcal{R}_{x_{\sigma(i)}}(z_i). \label{eq9}
\end{equation}

\para{Applying the LDP guarantee.}
Because $\mathcal{R}$ is $\varepsilon_0$-locally differentially private, its output distributions satisfy, for all users $t \in [n]\setminus\{1\}$ and all $z$ in the output domain,
\[
e^{-\varepsilon_0}
\;\le\;
\frac{\mathcal{R}_{x_1}(z)}{\mathcal{R}_{x_t}(z)}
\;\le\;
e^{\varepsilon_0}.
\]

Now fix $k$ and $k'$.  
Each permutation $\sigma$ contributing to Equation (\ref{eq8}) can be mapped one-to-one to a permutation $\sigma'$ in Equation (\ref{eq9}) that swaps the roles of users~1 and~$t=\sigma(k')$.  
This exchange changes at most two output terms in the product—those corresponding to $z_k$ and $z_{k'}$—each incurring at most an $e^{\varepsilon_0}$ multiplicative change by the LDP guarantee.  
Hence, for every such pair of permutations,
\[
\prod_{i=1}^n \mathcal{R}_{x_{\sigma(i)}}(z_i)
\;\le\;
e^{2\varepsilon_0}
\prod_{i=1}^n \mathcal{R}_{x_{\sigma'(i)}}(z_i).
\]
Summing over all $\sigma$ and using the bijection between $\sigma$ and $\sigma'$ gives
\[
\Pr[\boldsymbol{Z}=z\mid K=k,\boldsymbol{X}=\boldsymbol{x}]
\;\le\;
e^{2\varepsilon_0}\,\Pr[\boldsymbol{Z}=z\mid K=k',\boldsymbol{X}=\boldsymbol{x}].
\]
That is, the shuffled output distribution is $2\varepsilon_0$-DP with respect to the index $K$ given $\boldsymbol{X}$.

\para{Differential privacy and mutual information.}
It is known that if a random variable $Z$ satisfies $\varepsilon$-LDP with respect to a discrete input $K$, then
$I(K;Z)\le \varepsilon$ \cite{Cuff16} \footnote{There are more refined bounds like $\min\{\varepsilon^2,\varepsilon\}$ \cite{Cuff16}.}.
In our case, this yields
\[
I(K;\boldsymbol{Z}\mid \boldsymbol{X}) \le 2\varepsilon_0,
\]
completing the proof.
\end{proof}

\begin{corollary}\label{corollary_shuffle_IK}
In the shuffle-DP setting, the information leakage about the position $K$ is bounded by
\[
I(K;\boldsymbol{Z}) \le 2\varepsilon_0.
\]
\end{corollary}

\begin{proof}
Using the chain rule for mutual information,
\[
I(K;\boldsymbol{Z}\mid \boldsymbol{X})
= I(K;\boldsymbol{Z}) + I(K;\boldsymbol{X}\mid \boldsymbol{Z}) - I(K;\boldsymbol{X}).
\]
Since $K$ is independent of $\boldsymbol{X}$, we have $I(K;\boldsymbol{X})=0$.  
Hence,
\[
I(K;\boldsymbol{Z})
= I(K;\boldsymbol{Z}\mid \boldsymbol{X}) - I(K;\boldsymbol{X}\mid \boldsymbol{Z})
\le I(K;\boldsymbol{Z}\mid \boldsymbol{X})\le 2\varepsilon_0,
\]
because $I(K;\boldsymbol{X}\mid \boldsymbol{Z})\ge 0$ and Theorem \ref{theorem_shuffle_IK}.
\end{proof}

The same technique yields an information leakage bound for $I(K;\boldsymbol{Z})$ in the general shuffle-only setting:

\begin{theorem}\label{theorem_IK_general_shuffle_only}
In the shuffle-only setting, suppose the distributions $\{P_i\}_{i=1}^n$ satisfy
\[
\forall\, y\in\mathbb{Y},\; i,j\in[n],\qquad 
\frac{P_i(y)}{P_j(y)} \le e^{\varepsilon_0}.
\]
Then
\[
I(K;\boldsymbol{Z}) \le 2\varepsilon_0.
\]
\end{theorem}

\begin{proof}
Define a local randomizer $\mathcal{R}$ on $[n]$ by $\mathcal{R}(i)=P_i$.  
The shuffle-only model can then be viewed as a special case of the shuffle-DP setting with fixed inputs $\boldsymbol{X} = (1,2,\dots,n)$.  
Applying Theorem~\ref{theorem_shuffle_IK} directly yields the claimed bound.
\end{proof}

\subsection{Analysis of $I(X_1;\boldsymbol{Z}\mid \boldsymbol{X}_{-1})$}\label{sec_shuffle_dp_IX}

Analogous to the approach in Section~\ref{sec_general_shuffle_only}, the shuffle-DP setting can be reduced to a simpler and analytically tractable \emph{blanket-mix} model.  
This reduction relies on the following decomposition of the local randomizer $\mathcal{R}$.

\begin{definition}[Blanket decomposition of $\mathcal{R}$~\cite{Balle2019}]
Let $\gamma = \sum_{y \in \mathbb{Y}} \inf_x \mathcal{R}_x(y)$.  
The \emph{blanket distribution} $Q^{\mathrm{B}}$ of an $\varepsilon_0$-DP mechanism $\mathcal{R}$ is defined by
\[
Q^{\mathrm{B}}(y) \;=\; \frac{1}{\gamma}\inf_x \mathcal{R}_x(y).
\]
Each conditional distribution $\mathcal{R}(x)$ can then be written as a mixture
\[
\mathcal{R}(x) \;=\; \gamma\, Q^{\mathrm{B}} + (1-\gamma)\, \mathrm{LO}(x),
\]
where $\mathrm{LO}(x)$ is the \emph{left-over} distribution associated with input $x$.  
Equivalently, with probability $\gamma$ the mechanism outputs a draw from the common blanket $Q^{\mathrm{B}}$, and with probability $1-\gamma$ it outputs from the input-dependent residual $\mathrm{LO}(x)$.
\end{definition}

For technical convenience, we also consider the \emph{generalized blanket distribution} $\bar{Q}^{\mathrm{B}}$ of $\mathcal{R}$~\cite{su2025decompositionbasedoptimalboundsprivacy}:
\[
\bar{Q}^{\mathrm{B}}(y) =
\begin{cases}
\inf_{x} \mathcal{R}_x(y), & y \in \mathbb{Y}, \\[4pt]
1 - \gamma, & y = \perp,
\end{cases}
\]
where $\perp$ is a placeholder symbol representing ``non-blanket'' outcomes. The effective output space is thus extended to $\bar{\mathbb{Y}} = \mathbb{Y} \cup \{\perp\}$.

\begin{theorem}[Blanket reduction~\cite{su2025decompositionbasedoptimalboundsprivacy}]
\label{theorem_blanket2}
Let $\boldsymbol{Z}^{r} = \mathcal{S}(\{Y_1\} \cup \{Y_i' : i = 2, \dots, n\})$,  
where $Y_i' \stackrel{\text{i.i.d.}}{\sim} \bar{Q}^{\mathrm{B}}$ and $\mathcal{S}$ is a uniform random shuffler.  
Then there exists a post-processing function $f^{\mathrm{B}}$ such that
\[
\forall\, \boldsymbol{x}_{-1} \in \mathbb{X}^{n-1}:\quad
f^{\mathrm{B}}(\boldsymbol{Z}^{r}, \{\mathrm{LO}(x_i)\}_{i=2}^n) 
\stackrel{d}{=} 
\boldsymbol{Z} \mid (\boldsymbol{X}_{-1} = \boldsymbol{x}_{-1}),
\]
i.e., conditioned on $\boldsymbol{X}_{-1}$, the distribution of $\boldsymbol{Z}$ can be generated from $\boldsymbol{Z}^{r}$ via post-processing.
\end{theorem}

\begin{corollary}\label{corollary2}
By the data-processing inequality, post-processing cannot increase mutual information. Consequently,
\[
\forall\, \boldsymbol{x}_{-1} \in \mathbb{X}^{n-1}:\quad
I(X_1;\boldsymbol{Z}\mid \boldsymbol{X}_{-1}= \boldsymbol{x}_{-1})
\;\le\;
I(X_1;\boldsymbol{Z}^{r}\mid \boldsymbol{X}_{-1}= \boldsymbol{x}_{-1}).
\]
\end{corollary}

We now analyze $I(X_1;\boldsymbol{Z}^{r}\mid \boldsymbol{X}_{-1}= \boldsymbol{x}_{-1})$ in the blanket-mix model. When conditioning on $\boldsymbol{X}_{-1}= \boldsymbol{x}_{-1}$, it is sufficient to work with $X_1 \sim P$ where $P$ is the conditional probability distribution on $\boldsymbol{X}_{-1}= \boldsymbol{x}_{-1}$.

It is well known that producing a histogram of the outputs is equivalent (up to post-processing) to applying a random permutation to the multiset of messages; each representation can be obtained from the other~\cite{CheuThesis}.  
For clarity, we work with the histogram view in the following lemma.

\begin{lemma}[Mutual information for one signal mixed with $s$ i.i.d.\ $Q$-samples]\label{lem:fixed-s}
Let $\mathbb{Y}$ be a finite output alphabet, and assume $Q(y)>0$ for all $y\in\mathbb{Y}$ so that all likelihood ratios $\mathcal R_x(y)/Q(y)$ are well-defined and bounded.
For an integer $s\ge 1$, let
\[
M^{(s)} \sim \mathrm{Mult}(s;Q)
\]
denote the \emph{histogram} of $s$ independent $Q$-samples.  
Thus $M^{(s)}=(M^{(s)}_y)_{y\in\mathbb{Y}}\in\mathbb{N}^{|\mathbb{Y}|}$ satisfies
$\sum_{y}M^{(s)}_y=s$.
Assume $M^{(s)}$ is independent of the ``signal'' $(X_1,Y_1)$, where $X_1\sim P$ and
$Y_1\mid X_1=x \sim \mathcal R_x$.

\medskip
Define the \emph{mixed} $(s+1)$-sample histogram
\[
C^{(s)}:= e_{Y_1} + M^{(s)},
\]
where $e_{Y_1}\in\mathbb{N}^{|\mathbb{Y}|}$ is the one-hot vector indicating the single draw $Y_1$.
Define the corresponding chi-square divergences:
\[
\chi^2(\mathcal R_x\Vert Q)
:=\sum_y \frac{\mathcal R_x(y)^2}{Q(y)}-1,
\qquad
\overline{\chi^2}
:=\mathbb E_{x\sim P}\big[\chi^2(\mathcal R_x\Vert Q)\big].
\]
Let $\bar P(y)=\sum_x P(x)\mathcal R_x(y)$ be the $P$-mixture of the $\mathcal R_x$'s (i.e., the
marginal distribution of $Y_1$).

\medskip
Then the mutual information between $X_1$ and the histogram $C^{(s)}$ satisfies, for all sufficiently large $s$,
\begin{equation}\label{eq:fixed-s}
I(X_1; C^{(s)})
\;=\;
\frac{\overline{\chi^2}\;-\;\chi^2(\bar P\Vert Q)}{\,2(s+1)\,}
\;+\;
O\!\big((s+1)^{-3/2}\big).
\end{equation}

That is, when one input-dependent sample is mixed with $s$ independent $Q$-samples, the leakage about $X_1$ decays as $\Theta(1/s)$, and its leading coefficient is exactly the ``variance'' quantity $\overline{\chi^2}-\chi^2(\bar P\Vert Q)$.
\end{lemma}

\begin{proof}
Given $C^{(s)}=e_{Y_1}+M^{(s)}$, Bayes’ rule implies (Fact \ref{fact1} in Appendix \ref{appendix_3})
\[
\Pr(X_1=x\mid C^{(s)}) \;\propto\; P(x)\,T_x^{(s)},\qquad
T_x^{(s)}:=\sum_{y} C^{(s)}_y\,w_x(y).
\]
We write
\[
T_x^{(s)}=(s+1)+A_x^{(s)},\qquad
A_x^{(s)}:=\underbrace{\sum_y\big(M^{(s)}_y-sQ(y)\big)\,w_x(y)}_{\text{multinomial fluctuation}}
\;+\;\underbrace{\big(w_x(Y_1)-1\big)}_{\text{single-draw term}}.
\]
Define
\[
\varepsilon_x^{(s)}:=\frac{A_x^{(s)}}{s+1},\qquad
\bar\varepsilon^{(s)}:=\sum_u P(u)\,\varepsilon_u^{(s)},
\]
so that
\[
\pi(x\mid C^{(s)})=\frac{P(x)\,(1+\varepsilon_x^{(s)})}{1+\bar\varepsilon^{(s)}}.
\]

A discrete Taylor expansion of $D(\pi\Vert P)$ around $P$ yields (Lemma \ref{lemma3} in Appendix~\ref{appendix_3})
\begin{equation}\label{eq:PT}
D\big(\pi(\cdot\mid C^{(s)})\Vert P\big)
=\frac{1}{2}\sum_x P(x)\big(\varepsilon_x^{(s)}-\bar\varepsilon^{(s)}\big)^2
\;+\;R(C^{(s)}),
\qquad |R(C^{(s)})|\le K\,\|\varepsilon^{(s)}\|_3^3,
\end{equation}
where $\|\varepsilon^{(s)}\|_3^3:=\sum_x P(x)\,|\varepsilon_x^{(s)}|^3$ and $K$ is a constant.  
Taking expectations and using $I(X_1;C^{(s)})=\mathbb E[D(\pi\Vert P)]$, we obtain
\begin{equation}\label{eq:I-split}
I(X_1;C^{(s)})
=\frac{1}{2}\,\mathbb E\!\left[\sum_x P(x)\big(\varepsilon_x^{(s)}-\bar\varepsilon^{(s)}\big)^2\right]
+\mathbb E\,R(C^{(s)}).
\end{equation}

Since $\bar\varepsilon^{(s)}=\sum_x P(x)\varepsilon_x^{(s)}$,\footnote{It is important to clarify the meaning of the notation
\(\mathrm{Var}_{x\sim P}(\cdot)\) appearing in this expression.
Here the randomness is \emph{not} in the input variable \(X_1\).
Instead, once the histogram \(C^{(s)}\) is fixed, each quantity
\(A_x^{(s)}\) becomes a deterministic function of \(x\), and the distribution
\(P(x)\) is used only as a weighting measure on the index set \(\mathbb X\).
Thus
\[
\mathrm{Var}_{x\sim P}\!\big(A_x^{(s)}\big)
:=\sum_x P(x)\Big(A_x^{(s)}-\sum_u P(u)A_u^{(s)}\Big)^2
\]
is simply the variance of the \emph{function} \(x\mapsto A_x^{(s)}\) when the
index \(x\) is sampled from \(P\).  
With this interpretation, the identity above is just an algebraic rewriting of
\(\bar\varepsilon^{(s)}=\sum_x P(x)\varepsilon_x^{(s)}\), together with
\(\varepsilon_x^{(s)}=A_x^{(s)}/(s+1)\).}
\[
\sum_x P(x)\big(\varepsilon_x^{(s)}-\bar\varepsilon^{(s)}\big)^2
=\frac{1}{(s+1)^2}\,\mathrm{Var}_{x\sim P}\!\big(A_x^{(s)}\big),
\]
Hence
\begin{equation}\label{eq:I-var}
I(X_1;C^{(s)})
=\frac{1}{2(s+1)^2}\,\mathbb E\!\Big[\mathrm{Var}_{x\sim P}\!\big(A_x^{(s)}\big)\Big]
\;+\;\mathbb E\,R(C^{(s)}).
\end{equation}

According to Lemma \ref{lemma_main} in Appendix \ref{appendix_3}, $$\mathbb E\!\big[\mathrm{Var}_{x\sim P}(A_x^{(s)})\big]
=s\big(\overline{\chi^2}-\chi^2(\bar P\Vert Q)\big)\;+\;O(1).$$
From \eqref{eq:PT} and \eqref{eq:I-split}, it remains to bound $\mathbb E\|\varepsilon^{(s)}\|_3^3$.
Since $\varepsilon_x^{(s)}=A_x^{(s)}/(s+1)$ and $A_x^{(s)}$ is the sum of $s$ centered bounded multinomial increments plus the bounded term $w_x(Y_1)-1$, a standard moment bound (e.g., Rosenthal inequality) implies $\mathbb E|A_x^{(s)}|^3=O(s^{3/2})$ uniformly in $x$.  
Hence
\[
\mathbb E\|\varepsilon^{(s)}\|_3^3
=\sum_x P(x)\,\mathbb E\!\left[\left|\frac{A_x^{(s)}}{s+1}\right|^3\right]
=O\big((s+1)^{-3/2}\big),
\]
and therefore $\mathbb E\,R(C^{(s)})=O((s+1)^{-3/2})$.

Substituting \eqref{eq:VarA-final} into \eqref{eq:I-var} yields
\[
I(X_1;C^{(s)})
=\frac{1}{2(s+1)^2}\Big\{s\big(\overline{\chi^2}-\chi^2(\bar P\Vert Q)\big)+O(1)\Big\}
+O\big((s+1)^{-3/2}\big)
=\frac{\overline{\chi^2}-\chi^2(\bar P\Vert Q)}{2(s+1)}+O\big((s+1)^{-3/2}\big),
\]
which is precisely \eqref{eq:fixed-s}.
\end{proof}

\begin{theorem}\label{theorem_shuffle_IX}
In the shuffle-DP setting, for any fixed $\boldsymbol{x}_{-1} \in \mathbb{X}^{n-1}$,
\[
I\!\bigl(X_1;\boldsymbol{Z}\mid \boldsymbol{X}_{-1}= \boldsymbol{x}_{-1}\bigr)
\;\le\;
\frac{e^{\varepsilon_0}-1}{2n}
+ O(n^{-3/2}).
\]
Taking expectation over $\boldsymbol{X}_{-1}$ yields
\[
I(X_1;\boldsymbol{Z}\mid \boldsymbol{X}_{-1})
\;\le\;
\frac{e^{\varepsilon_0}-1}{2n}
+ O(n^{-3/2}).
\]
\end{theorem}

\begin{proof}
By the definition of the generalized blanket distribution $\bar{Q}^{\mathrm{B}}$ and the $\varepsilon_0$-DP property of $\mathcal{R}$,
\[
\forall x\in\mathbb{X}, \quad
\chi^2(\mathcal{R}_x \Vert \bar{Q}^{\mathrm{B}})
= \sum_y \mathcal{R}_x(y)\frac{\mathcal{R}_x(y)}{\bar{Q}^{\mathrm{B}}(y)} - 1
\le \sum_y \mathcal{R}_x(y)e^{\varepsilon_0} - 1
= e^{\varepsilon_0}-1,
\]
Hence 
\[
\overline{\chi^2} - \chi^2(\bar P \Vert Q)
\;\le\;
\overline{\chi^2}=\mathbb{E}_x[\chi^2(\mathcal{R}_{x} \Vert \bar{Q}^{\mathrm{B}})]
\;\le\;
e^{\varepsilon_0}-1.
\]
Applying Lemma~\ref{lem:fixed-s} with $s=n-1$ and then using Corollary~\ref{corollary2} gives
\[
I\!\bigl(X_1;\boldsymbol{Z}\mid \boldsymbol{X}_{-1}=\boldsymbol{x}_{-1}\bigr)
\;\le\;
I\!\bigl(X_1;\boldsymbol{Z}^{r}\mid \boldsymbol{X}_{-1}=\boldsymbol{x}_{-1}\bigr)
\;\le\;
\frac{e^{\varepsilon_0}-1}{2n}
+ O(n^{-3/2}).\qedhere
\]
\end{proof}

\begin{corollary}\label{corollary_shuffle_IX}
In the shuffle-DP setting, if $X_1$ is independent of $\boldsymbol{X}_{-1}$, then
\[
I(X_1;\boldsymbol{Z})
\;\le\;
\frac{e^{\varepsilon_0}-1}{2n}
+ O(n^{-3/2}).
\]
\end{corollary}

\begin{proof}
The argument follows the same reasoning as in Corollary~\ref{corollary_shuffle_IK}.
\end{proof}

In fact, Lemma~\ref{lem:fixed-s} also yields \emph{lower bounds} on the mutual information in the shuffle-DP setting.  
When \(X_2, X_3, \dots, X_n\) are i.i.d., the output distribution \(Q = \mathcal{R}(X_2)\) naturally plays the role of an \emph{effective blanket distribution}.  
This interpretation is illustrated by the following example.

\begin{example}\label{example_rr}
Let \( [k] := \{1,2,\dots,k\} \) and denote by \( \mathcal{U}_{[k]} \) the uniform distribution on \( [k] \). 
For any \( k \in \mathbb{N} \) and \( \varepsilon_0 > 0 \), the \(k\)-ary randomized response mechanism 
\( k\mathrm{RR}_{\varepsilon_0}: [k] \to [k] \) is defined as
\[
k\mathrm{RR}_{\varepsilon_0}(x) =
\begin{cases}
x, & \text{with probability } \displaystyle \frac{e^{\varepsilon_0}-1}{e^{\varepsilon_0} + k - 1}, \\[4pt]
y \sim \mathcal{U}_{[k]}, & \text{with probability } \displaystyle \frac{k}{e^{\varepsilon_0} + k - 1}.
\end{cases}
\]
The generalized blanket distribution of \(k\mathrm{RR}_{\varepsilon_0}\) is
\[
\bar{Q}^{\mathrm{B}}
= \frac{k}{e^{\varepsilon_0} + k - 1}\,\mathcal{U}_{[k]}
+ \frac{e^{\varepsilon_0} - 1}{e^{\varepsilon_0} + k - 1}\,\mathbbm{1}_{\perp}.
\]
For any input distribution \(P\),
\[
\overline{\chi^2}(P\Vert \bar{Q}^{\mathrm{B}})
:= \mathbb{E}_{x\sim P}\!\big[\chi^2(\mathcal{R}_x\Vert \bar{Q}^{\mathrm{B}})\big]
= \frac{e^{\varepsilon_0}}{e^{\varepsilon_0}+k-1}\,(e^{\varepsilon_0}-1).
\]
By Lemma~\ref{lem:fixed-s}, this implies the upper bound
\[
\forall \boldsymbol{x}_{-1}:
I(X_1;\boldsymbol{Z}\mid \boldsymbol{X}_{-1}= \boldsymbol{x}_{-1})
\;\le\;
\frac{\overline{\chi^2}(P\Vert \bar{Q}^{\mathrm{B}})}{2n}
+ O(n^{-3/2})
=
\frac{e^{\varepsilon_0}}{e^{\varepsilon_0}+k-1}
\cdot
\frac{e^{\varepsilon_0}-1}{2n}
+ O(n^{-3/2}).
\]

Now suppose that \(X_2, X_3, \dots, X_n\) are i.i.d.\ uniform on \([k]\).  
Then the actual blanket distribution is \( Q = k\mathrm{RR}_{\varepsilon_0}(\mathcal{U}_{[k]})=\mathcal{U}_{[k]} \). 
In this case, for any \(P\),
\[
\overline{\chi^2}(P\Vert Q)
:= \mathbb{E}_{x\sim P}\!\big[\chi^2(\mathcal{R}_x\Vert Q)\big]
= \frac{(k-1)(e^{\varepsilon_0}-1)^2}{(e^{\varepsilon_0}+k-1)^2}.
\]
When $X_1\sim P= \mathcal{U}_{[k]}$, the mixture $\bar{P}$ coincides with $Q=\mathcal{U}_{[k]}$, so $\chi^2(\bar{P}\Vert Q)=0$.
Applying Lemma~\ref{lem:fixed-s} again yields
\[
I(X_1;\boldsymbol{Z})
= \frac{(k-1)(e^{\varepsilon_0}-1)^2}{2(e^{\varepsilon_0}+k-1)^2\,n}
+ O(n^{-3/2}),
\]
showing that there exist input distributions $(X_i)$ for which this rate is asymptotically attained.

Finally, we run a numerical experiment in the shuffle-DP setting using the 4-ary randomized response mechanism with uniformly distributed inputs \(X_i\) \((i=1,2,\dots,n)\). 
The empirical mutual information \(I(X_1;\boldsymbol{Z})\), estimated via Monte Carlo with \(10^{5}\) samples, closely matches the asymptotic prediction
\(\tfrac{\overline{\chi^2}(P\Vert Q)-\chi^2(\bar{P}\Vert Q)}{2n}=\frac{\overline{\chi^2}(P\Vert Q)}{2n}\),
and lies below both the blanket-decomposition upper bound
\(\tfrac{\overline{\chi^2}(P\Vert \bar{Q}^{\mathrm{B}})}{2n}\)
and the unified bound \(\tfrac{e^{\varepsilon_0}-1}{2n}\).
These results provide empirical support for our asymptotic analysis.

\begin{figure}[t]
    \centering
    \includegraphics[width=0.5\linewidth]{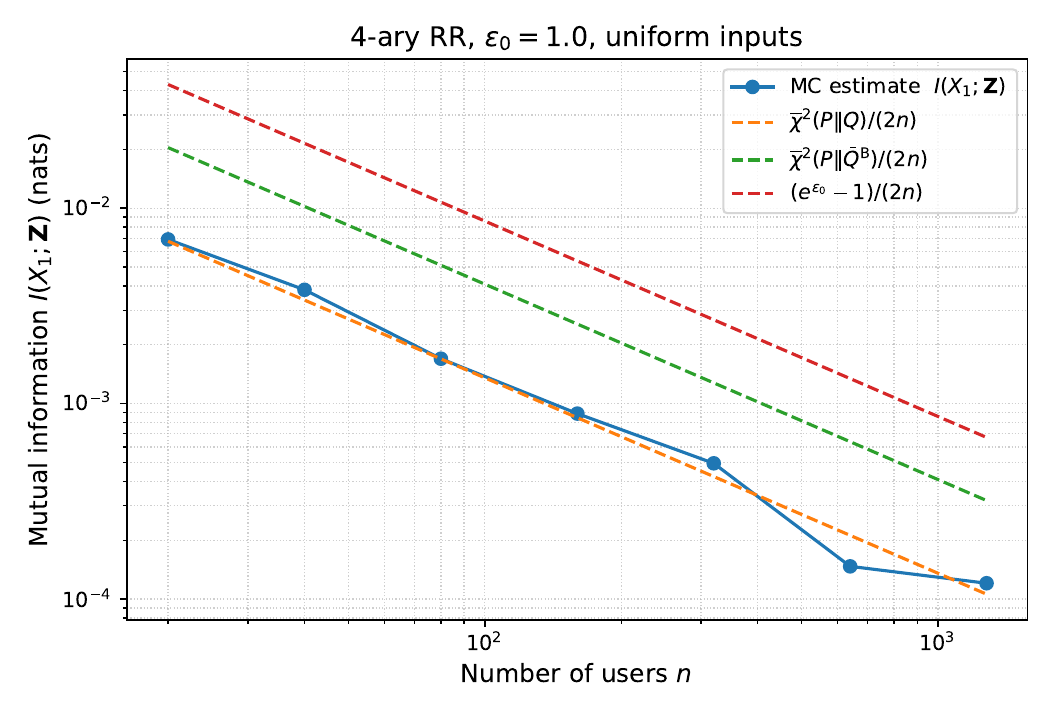}
    \caption{Mutual information in the shuffle-DP setting: numerical estimates vs.\ asymptotic bounds.}
    \label{fig:fig3}
\end{figure}
\end{example}

\section{Discussion}\label{sec_discussion}
We employ the blanket decomposition to reduce the intractable shuffle-DP channel to a more analyzable ``blanket–mix'' form.  
Another well-known decomposition is the \emph{clone decomposition}~\cite{Feldman2021}, which states that for any fixed $x_1$, the output distribution of any other input $x$ can be written as
\[
\forall x\in\mathbb{X},\qquad
\mathcal{R}(x)
= e^{-\varepsilon_0}\,\mathcal{R}(x_1)
+ \bigl(1-e^{-\varepsilon_0}\bigr)\,\mathrm{LO}(x).
\]
Under this representation, with probability $e^{-\varepsilon_0}$ an arbitrary user effectively produces a ``clone'' of the user holding $x_1$.

However, the clone decomposition is \emph{not} suitable for analyzing
\(
I(X_1;\boldsymbol{Z}\mid \boldsymbol{X}_{-1}).
\)
Indeed, after conditioning on $\boldsymbol{X}_{-1}$, the decomposition reduces the shuffled multiset to
\[
1+N \text{ samples from }\mathcal{R}(X_1),
\qquad N\sim\mathrm{Bin}(n-1,e^{-\varepsilon_0}),
\]
which, with overwhelming probability, reveals almost complete information about $X_1$ (see Remark~\ref{remark3}).  
In contrast, the blanket distribution is \emph{independent} of $X_1$, which preserves the anonymity needed for a meaningful information-theoretic analysis of $I(X_1;\boldsymbol{Z})$.

\medskip
Nonetheless, the clone decomposition remains very useful for understanding
\(I(Y_1;\boldsymbol{Z}\mid \boldsymbol{X}_{-1})\).
When our goal is to bound the leakage about the \emph{message} $Y_1$ rather than the underlying input $X_1$, the clone representation reduces the shuffle-DP mechanism to the \emph{basic shuffle-only model} (Section \ref{sec_basic_PeqQ}) with
\[
P = Q
\qquad\text{and}\qquad 
N = 1 + \mathrm{Bin}(n-1,e^{-\varepsilon_0}).
\]
Applying Theorem~\ref{theorem2} to this reduced instance gives
\[
I(Y_1;\boldsymbol{Z}\mid \boldsymbol{X}_{-1})
\;\le\;
\mathbb{E}_{N\sim\mathrm{Bin}(n-1,e^{-\varepsilon_0})}
\!\left[\frac{m-1}{2(N+1)}\right]
+ O(n^{-3/2})
\;=\;
\frac{(m-1)e^{\varepsilon_0}}{2n}
+ O(n^{-3/2}),
\]
where \(m = |\mathbb{Y}|\).
Thus the clone decomposition, although unsuitable for bounding the leakage about \emph{inputs}, provides a clean characterization of message-level privacy leakage.

It is worth noting that the blanket decomposition also applies to the analysis of 
$I(Y_1;\boldsymbol{Z}\mid \boldsymbol{X}_{-1})$, in which case the corresponding reduced setting is the basic 
shuffle-only model with $P\neq Q$.  
As shown in~\cite{su2025decompositionbasedoptimalboundsprivacy}, conditioned on $\boldsymbol{X}_{-1}$, there exists a 
post-processing function from the mixture induced by the clone reduction to that induced by the 
blanket reduction.  
Consequently, the blanket decomposition yields a tighter upper bound, although the resulting 
expressions are algebraically more involved than those obtained from the clone decomposition.

\section{Conclusion}\label{sec_conclusion}
This paper presents the first systematic information-theoretic study of the \emph{single-message shuffle model}. 
We analyze both the shuffle-only and shuffle-DP settings, deriving asymptotic expressions for the mutual information between the shuffled output and either the target user’s message position \(I(K;\boldsymbol{Z})\) or the target user’s input \(I(X_1;\boldsymbol{Z})\). 
Our results reveal how anonymization (through shuffling) and randomization (through local differential privacy) jointly limit information leakage. 
From a methodological standpoint, we extend the \emph{blanket decomposition} technique to the mutual-information setting, providing a unified analytical framework for quantifying privacy in multi-user systems. 
This framework may facilitate further research on information-theoretic approaches to the shuffle model and its extensions.

\medskip
We conclude by highlighting two open problems that may inspire future research:
\begin{itemize}
    \item \textbf{Alternative analyses for the general shuffle-only setting.}
    The general shuffle-only scenario remains both intriguing and technically challenging. 
    In this work, we employ the \emph{blanket decomposition} technique to reduce the analysis to a simpler, tractable form. 
    This reduction relies on the existence of a common component among the remaining users' distributions to obscure the target user. 
    An open question is whether alternative analytical frameworks can be developed to study this setting directly without relying on such a decomposition.
    
    \item \textbf{Tight and non-asymptotic bounds on information leakage.}
    The residual terms in our asymptotic bounds depend on the distributions \(P, Q\) in the shuffle-only setting and on the local randomizer \(\mathcal{R}\) in the shuffle-DP setting. 
    This dependence appears inevitable in the shuffle-only case, as it is inherently related to the support size of \(P\) and \(Q\) (see Theorem~\ref{theorem2}). 
    However, for the shuffle-DP regime, it remains unclear whether one can establish a \emph{non-asymptotic}, distribution-independent upper bound on \(I(X_1;\boldsymbol{Z}\mid \boldsymbol{X}_{-1})\). 
    Resolving this question would offer a deeper information-theoretic understanding of the shuffle model and its privacy amplification behavior.
\end{itemize}

\bibliographystyle{IEEEtran}
\bibliography{shuffle}

{\appendices
\section{Proof of Theorem \ref{theorem_posterior}}\label{appendix1}
\begin{proof}
Recall that the shuffling process permutes the users’ messages according to a random permutation $\sigma$, uniformly sampled from the permutation group $\mathcal{S}_n$.  
Let $\boldsymbol{y}=(y_1,\dots,y_n)$ denote the original (unshuffled) message sequence.  
By definition, the observed shuffled output satisfies $z_i=y_{\sigma(i)}$ for each $i$.

The posterior probability of $K=k$ given $\boldsymbol{Z}=\boldsymbol{z}$ can be written as
\begin{align}
\Pr[K = k \mid \boldsymbol{Z}=\boldsymbol{z}]
&= \Pr[K = k \mid \boldsymbol{y} = \boldsymbol{z}_{\sigma^{-1}}], \label{eq:posterior_start}
\end{align}
since knowing $\boldsymbol{Z}=\boldsymbol{z}$ is equivalent to knowing $\boldsymbol{y}$ under the inverse permutation.

Using the definition of conditional probability,
\begin{align}
\Pr[K = k \mid \boldsymbol{y} = \boldsymbol{z}_{\sigma^{-1}}]
&= \frac{\Pr[K = k \wedge \boldsymbol{y} = \boldsymbol{z}_{\sigma^{-1}}]}{\Pr[\boldsymbol{y} = \boldsymbol{z}_{\sigma^{-1}}]}. \label{eq:bayes}
\end{align}
We next express the numerator and denominator using the law of total probability over all possible permutations $\sigma_1 \in \mathcal{S}_n$:
\begin{align}
\Pr[K = k \mid \boldsymbol{Z}=\boldsymbol{z}]
&= \frac{\Pr[K = k] \cdot \Pr[\boldsymbol{y} = \boldsymbol{z}_{\sigma^{-1}} \mid K = k]}{\sum_{\sigma_1 \in \mathcal{S}_n} \Pr[\sigma = \sigma_1] \cdot \Pr[\boldsymbol{y} = \boldsymbol{z}_{\sigma^{-1}} \mid \sigma = \sigma_1]}. \label{eq:law_total}
\end{align}

Two key facts are used here:
\begin{enumerate}
\item The position $K$ is uniformly distributed, i.e., $\Pr[K=k]=1/n$.
\item Conditioned on $K=k$, the element at position $z_k$ is drawn from $P$, whereas the remaining $z_j$’s are drawn from $Q$.
\end{enumerate}
Hence, the joint likelihood becomes
\begin{align}
\Pr[\boldsymbol{y} = \boldsymbol{z}_{\sigma^{-1}} \mid K = k]
&= P(z_k)\prod_{j\neq k} Q(z_j)
= \frac{P(z_k)}{Q(z_k)}\prod_{j=1}^n Q(z_j). \label{eq:likelihood}
\end{align}
Plugging~\eqref{eq:likelihood} into~\eqref{eq:law_total}, we obtain
\begin{align}
\Pr[K = k \mid \boldsymbol{Z}=\boldsymbol{z}]
&= \frac{\frac{1}{n}\frac{P(z_k)}{Q(z_k)}\prod_{j=1}^n Q(z_j)}
{\sum_{i=1}^n \frac{1}{n}\frac{P(z_i)}{Q(z_i)}\prod_{j=1}^n Q(z_j)}. \label{eq:posterior_fraction}
\end{align}
After canceling the common factors $\frac{1}{n}$ and $\prod_{j=1}^n Q(z_j)$ in the numerator and denominator, we obtain the normalized posterior:
\begin{align*}
\Pr[K = k \mid \boldsymbol{Z}=\boldsymbol{z}]
&= \frac{\frac{P(z_k)}{Q(z_k)}}{\sum_{i=1}^n \frac{P(z_i)}{Q(z_i)}}. \qedhere
\end{align*}
\end{proof}

\section{Details in the Proof of Theorem \ref{theorem2}}\label{appendix_2}

This section clarifies several technical steps used in the proof.  
Let \(f(t) := t\log t\) and fix an index \(i\). Define
\[
g(t) \;=\; p_i \log p_i \;+\; (\log p_i + 1)\,(t - p_i) \;+\; \frac{1}{2p_i}\,(t - p_i)^2 .
\]
By the second-order Taylor expansion of \(f\) around \(t=p_i\) with Lagrange remainder, we have
\begin{equation}\label{eq:taylor2}
f(t) \;=\; g(t) \;+\; R(t),
\end{equation}
where 
\[
R(t) \;=\; \frac{f^{(3)}(\theta)}{6}\,(t - p_i)^3 
\;=\; -\frac{1}{6\theta^2}\,(t - p_i)^3, 
\qquad \theta \in (t, p_i).
\]
As \(t \to 0\), no uniform bound on \(1/\theta^2\) is available.  
To handle this rigorously, define the event
\[
\mathcal{E} \;:=\; \Bigl\{\tfrac{X}{n} \le \tfrac{p_i}{2}\Bigr\}.
\]
By the Chernoff bound, \(\Pr[\mathcal{E}] = \mathrm{negl}(n)\),\footnote{We use the standard notation \(\phi(n)=\mathrm{negl}(n)\) to mean that for every \(t\in\mathbb{N}\), \(\lim_{n\to\infty} n^t \phi(n)=0\).} i.e., the probability of \(\mathcal{E}\) decays faster than any inverse polynomial in \(n\).

For \(t \ge p_i/2\), the remainder is uniformly bounded as
\[
|R(t)| \;\le\; \frac{2}{3p_i^2}\,|t - p_i|^3 .
\]
Therefore,
\begin{align*}
\mathbb{E}_{X \sim \mathrm{Bin}(n,p_i)}\!\Bigl[f\!\Bigl(\tfrac{X}{n}\Bigr)\Bigr]
&= \mathbb{E}\!\Bigl[g\!\Bigl(\tfrac{X}{n}\Bigr)\Bigr] + \mathbb{E}\!\Bigl[R\!\Bigl(\tfrac{X}{n}\Bigr)\Bigr] \\
&= p_i \log p_i + \frac{1-p_i}{2n}
   + \sum_{x>np_i/2} p_X(x)\, R\!\Bigl(\tfrac{x}{n}\Bigr)
   + \sum_{x\le np_i/2} p_X(x)\, R\!\Bigl(\tfrac{x}{n}\Bigr) \\
&= p_i \log p_i + \frac{1-p_i}{2n}
   + O\!\Bigl(\sum_{x>np_i/2} p_X(x)\,\Bigl|R\!\Bigl(\tfrac{x}{n}\Bigr)\Bigr|\Bigr)
   + O\!\Bigl(\sum_{x\le np_i/2} p_X(x)\,\Bigl| R\!\Bigl(\tfrac{x}{n}\Bigr)\Bigr|\Bigr) \\
&= p_i \log p_i + \frac{1-p_i}{2n}
   + O\!\Bigl(\mathbb{E}\bigl[\,\bigl|\tfrac{X}{n}-p_i\bigr|^3\,\bigr]\Bigr)
   + \mathrm{negl}(n)\cdot M.
\end{align*}
where $M=\max_{t\in[0,1]}|R(t)|=\max_{t\in[0,1]}|f(t)-g(t)|=O(1)$.

To bound \(O\!\bigl(\mathbb{E}\bigl[\,|\tfrac{X}{n}-p_i|^3\,\bigr]\bigr)\), we use the following inequality.

\begin{theorem}[Rosenthal inequality~\cite{CS20}]\label{theorem_Rosenthal_ineq}
Let \(X_1,\dots,X_n\) be independent random variables with \(\mathbb{E}|X_i|^{p}<\infty\) for some \(p\ge 2\), and set \(Y_i := X_i - \mathbb{E}X_i\). Then there exists a constant \(C_p>0\) depending only on \(p\) such that
\[
\mathbb{E}\Bigl|\sum_{i=1}^n Y_i\Bigr|^p
\;\le\;
C_p\max\Bigl\{ \,\Bigl(\sum_{i=1}^n \mathbb{E}|Y_i|^2\Bigr)^{p/2}, \sum_{i=1}^n \mathbb{E}|Y_i|^{p} \Bigr\}.
\]
Consequently, there is \(C_p'>0\) depending only on \(p\) such that
\[
\mathbb{E}\!\left|\sum_{i=1}^n X_i\right|^{p}
\;\le\;
C_p'\!\left[
\left(\sum_{i=1}^n \mathrm{Var}[X_i]  \right)^{p/2}
+\sum_{i=1}^n \mathbb{E}|X_i-\mathbb{E}X_i|^{p}
+
\left|\sum_{i=1}^n \mathbb{E}X_i\right|^{p}
\right].
\]
\end{theorem}

Let \(A_1,\dots,A_n\) be i.i.d.\ Bernoulli random variables with mean \(p_i\), and write \(X=\sum_{j=1}^n A_j\). Then
\[
\left|\tfrac{X}{n} - p_i\right| \;=\; \frac{1}{n}\,\Bigl|\sum_{j=1}^n (A_j - p_i)\Bigr|.
\]
Applying Theorem~\ref{theorem_Rosenthal_ineq} with \(p=3\),
\begin{align*}
\mathbb{E}\Bigl|\sum_{j=1}^n (A_j - p_i)\Bigr|^3
\;&=\;
O\!\left(\, \Bigl(\sum_{j=1}^n \mathbb{E}[(A_j - p_i)^2]\Bigr)^{3/2}+ \sum_{j=1}^n \mathbb{E}|A_j-p_i|^3     \right)
\;\\
&=\;
O\!\left((n\,p_i(1-p_i))^{3/2}+n[p_i(1-p_i)^3+(1-p_i)p_i^3]  \right)=O(n^{3/2}).    
\end{align*}

Hence
\[
\mathbb{E}\Bigl|\tfrac{X}{n} - p_i\Bigr|^3
\;=\; \frac{1}{n^3}\,\mathbb{E}\Bigl|\sum_{j=1}^n (A_j - p_i)\Bigr|^3
\;=\; O\!\left(n^{-3/2}\right).
\]
Combining the bounds above yields the desired estimate.

\section{Details in the Proof of Theorem \ref{theorem_IK_general1}}\label{appendix_11}

\begin{lemma}\label{lem:ratio-delta}
Let \((S,T)\) be random variables (depending on the parameter \(n\)) with \(S>0\) a.s., and let
\[
\mu_S:=\mathbb{E}[S], \qquad \mu_T:=\mathbb{E}[T].
\]
Assume there exist constants \(c_1,c_2,c_3,c_4,c_5>0\), independent of \(n\), such that
\[
c_1 \le S \le c_2 n,\qquad \Pr[S\le c_3 n]=\mathrm{negl}(n),\qquad |T|\le c_4 n,\qquad 
\mathbb{E}\bigl(|S-\mu_S|^3+|T-\mu_T|^3\bigr)\le c_5 n^{3/2}.
\]
Let \(f(s,t):=t/s\). Then
\begin{equation}\label{eq:ratio-expansion2}
\mathbb{E}\!\left[\frac{T}{S}\right]
= \frac{\mu_T}{\mu_S}\;-\;\frac{\mathrm{Cov}(T,S)}{\mu_S^2}
\;+\;\frac{\mu_T\,\mathrm{Var}(S)}{\mu_S^3}\;+\;O(n^{-3/2}).
\end{equation}
\end{lemma}

\begin{proof}
Set
\[
\Delta_S:=S-\mu_S,\qquad \Delta_T:=T-\mu_T,
\]
so that \(\mathbb{E}[\Delta_S]=\mathbb{E}[\Delta_T]=0\),
\(\mathrm{Var}(S)=\mathbb{E}[\Delta_S^2]\),
\(\mathrm{Var}(T)=\mathbb{E}[\Delta_T^2]\), and
\(\mathrm{Cov}(S,T)=\mathbb{E}[\Delta_S\Delta_T]\).

For \(f(s,t)=t/s\) on \((0,\infty)\times\mathbb{R}\),
\[
f_s(s,t)=-\frac{t}{s^2},\quad f_t(s,t)=\frac{1}{s},\quad
f_{ss}(s,t)=\frac{2t}{s^3},\quad f_{st}(s,t)=-\frac{1}{s^2},\quad f_{tt}(s,t)=0.
\]
Consider the quadratic Taylor polynomial at \((\mu_S,\mu_T)\):
\[
g(S,T)
= f(\mu_S,\mu_T)
+ f_s(\mu_S,\mu_T)\,\Delta_S
+ f_t(\mu_S,\mu_T)\,\Delta_T
+ \frac{1}{2}\!\left(
f_{ss}(\mu_S,\mu_T)\,\Delta_S^2
+ 2 f_{st}(\mu_S,\mu_T)\,\Delta_S\Delta_T
+ f_{tt}(\mu_S,\mu_T)\,\Delta_T^2\right).
\]
By the two-variable Taylor formula with Lagrange remainder,
\begin{equation}\label{eq:taylor3}
f(S,T) \;=\; g(S,T) + R,
\end{equation}
where the third-order remainder satisfies
\begin{equation}\label{eq:R-bound}
|R|\;\le\; \frac{1}{6}\,\Big(\sup_{\mathcal{N}} \|D^{3} f\|\Big)\,
\bigl(|\Delta_S|+|\Delta_T|\bigr)^3,
\end{equation}
with the supremum taken over the line segment
\(\mathcal{N}=\{(\mu_S,\mu_T)+\theta(\Delta_S,\Delta_T):\theta\in[0,1]\}\), and \(\|D^{3} f\|\) any operator norm of the array of third partials. The nonzero third partials are \(f_{sss}(s,t)=-6t/s^{4}\) and \(f_{sst}(s,t)=2/s^{3}\).

From \eqref{eq:taylor3} and centering, the linear terms vanish, giving
\begin{align*}
\mathbb{E}\!\left[\frac{T}{S}\right]
&= f(\mu_S,\mu_T) 
+ \frac{1}{2}\!\left(
f_{ss}(\mu_S,\mu_T)\,\mathrm{Var}(S)
+ 2 f_{st}(\mu_S,\mu_T)\,\mathrm{Cov}(S,T)
+ f_{tt}(\mu_S,\mu_T)\,\mathrm{Var}(T)\right)
+ \mathbb{E}[R] \\
&= \frac{\mu_T}{\mu_S}
- \frac{\mathrm{Cov}(T,S)}{\mu_S^{2}}
+ \frac{\mu_T\,\mathrm{Var}(S)}{\mu_S^{3}}
+ \mathbb{E}[R].
\end{align*}

\para{Bounding the remainder \(\mathbb{E}[R]\).}
Let \(\mathcal{E}=\{S\ge c_3 n\}\). On \(\mathcal{E}\), since \(|T|\le c_4 n\) and \(S\ge c_3n\), the nonzero third partial derivatives of \(f(s,t)=t/s\) satisfy
\[
\sup_{\mathcal{N}}\|D^{3}f\|
=\max\!\left\{\sup_{\mathcal{N}}\bigl|f_{sss}\bigr|,\ \sup_{\mathcal{N}}\bigl|f_{sst}\bigr|\right\}
\le \max\{\frac{6c_4}{c_3^4n^3},\frac{2}{c_3^3n^3} \} =O(n^{-3}).
\]
Hence, by \eqref{eq:R-bound},
\[
\mathbb{E}\bigl[|R|\mathbf{1}_{\mathcal{E}}\bigr]
\;\le\; O(n^{-3})\,\mathbb{E}\bigl[(|\Delta_S|+|\Delta_T|)^3\bigr]
\;=\; O(n^{-3})\cdot O(n^{3/2})
\;=\; O(n^{-3/2}),
\]
using the moment assumption.

On \(\mathcal{E}^{c}\), we use a crude polynomial bound. Since \(S\ge c_1>0\) and \(|T|\le c_4 n\), we have
\begin{align*} |f(S,T)|&\le \frac{c_4}{c_1}n,\\ |g(S,T)|&\le |f(\mu_S,\mu_T)| + |f_s(\mu_S,\mu_T)\,\Delta_S| + |f_t(\mu_S,\mu_T)\,\Delta_T| + \frac{1}{2}\!\left( |f_{ss}(\mu_S,\mu_T)\,\Delta_S^2| + 2 |f_{st}(\mu_S,\mu_T)\,\Delta_S\Delta_T| + |f_{tt}(\mu_S,\mu_T)\,\Delta_T^2|\right)\\ &\le \frac{c_4}{c_1}n +\frac{c_4n}{c_1^2}\cdot c_2n+\frac{c_4n}{c_1}+\frac{1}{2}(\frac{2c_4n}{c_1^3}\cdot (c_2n)^2+\frac{2}{c_1^2}\cdot c_2c_4n^2)\\ &=O(n^3). \end{align*}
Consequently,
\[
\mathbb{E}\bigl[|R|\mathbf{1}_{\mathcal{E}^{c}}\bigr]
\;\le\; \mathbb{E}\bigl[(|f(S,T)|+|g(S,T)|)\mathbf{1}_{\mathcal{E}^{c}}\bigr]
\;\le\; O(n^3)\,\Pr[\mathcal{E}^{c}]
\;=\; \mathrm{negl}(n).
\]
Combining the two parts yields \(\mathbb{E}[R]=O(n^{-3/2})\), which establishes \eqref{eq:ratio-expansion2}.
\end{proof}

It remains to verify that the assumptions of Lemma~\ref{lem:ratio-delta} hold for the random variables
\[
S:=\sum_{j=1}^n w(Z_j),\qquad 
T:=\sum_{j=1}^n w(Z_j)\log w(Z_j),\qquad 
\boldsymbol{Z}=(Z_1,\dots,Z_n).
\]
Among the \(n\) summands in \(S\), exactly one term is \(W^{(P)}:=w(Y_1)\) with \(Y_1\sim P\), while the remaining \(n-1\) terms are i.i.d.\ \(W_i^{(Q)}:=w(Y_i)\) with \(Y_i\sim Q\).

Let
\[
w_{\min}^{(P)}:=\min_{y\in \mathrm{Supp}(P)} \frac{P(y)}{Q(y)} \;>\;0,\quad
w_{\max}^{(P)}:=\max_{y\in \mathrm{Supp}(P)} \frac{P(y)}{Q(y)},\quad
w_{\max}^{(Q)}:=\max_{y\in \mathrm{Supp}(Q)} \frac{P(y)}{Q(y)},\quad
M:=\max_{y} \bigl|w(y)\log w(y)\bigr|.
\]
Then, deterministically,
\[
S \;\ge\; w_{\min}^{(P)},\qquad 
S \;\le\; w_{\max}^{(P)} + (n-1)\,w_{\max}^{(Q)} \;\le\; n\cdot \max\{w_{\max}^{(P)},w_{\max}^{(Q)}\},
\qquad
|T| \;\le\; nM.
\]
Moreover,
\[
\mu_S:=\mathbb{E}[S] \;=\; \mathbb{E}_P[W] + (n-1)\mathbb{E}_Q[W] \;=\; (1+\chi^2) + (n-1)\cdot 1 \;=\; n+\chi^2,
\]
so \(\mu_S=\Theta(n)\). Since \(\{w(Z_j)\}_{j=1}^n\) are bounded and independent, a standard Chernoff bound yields, for some constant \(c>0\),
\[
\Pr\!\left[S \le \tfrac{1}{2}\mu_S\right] \;\le\; e^{-c n}\;=\;\mathrm{negl}(n).
\]
For the third centered moments, by the Rosenthal inequality (with \(p=3\)),
\[
\mathbb{E}\bigl(|S-\mu_S|^3\bigr) \;
\;=\; O\!\bigl(n^{3/2}\bigr),
\qquad
\mathbb{E}\bigl(|T-\mu_T|^3\bigr) \;
\;=\; O\!\bigl(n^{3/2}\bigr),
\]
since \(\mathrm{Var}\,S=\Theta(n)\) and \(\mathrm{Var}\,T=\Theta(n)\). These bounds verify the hypotheses of Lemma~\ref{lem:ratio-delta} with constants independent of \(n\), completing the check.

\section{The post-processing function of Blanket decomposition}
The corresponding algorithm is shown in Algorithm \ref{alg:post-processing2}.
\begin{algorithm}[t]
\caption{Post-processing function of privacy blanket, $f^{\mathsf{B}}$ }
\label{alg:post-processing2}
\begin{algorithmic}[0]
\State \textbf{Meta parameter:} $\text{LO}_j,j=2,3,\dots,n$ 
\Require $\boldsymbol{Z}^r$
\State $\boldsymbol{Z} \gets \boldsymbol{Z}^r$
\State $J \gets \emptyset$
\State $S \gets \emptyset$
\For{$i = 1, \dots, n$}
\If{$\boldsymbol{Z}^r[i]=\perp$}
    \State Let $j$ be a randomly and uniformly chosen element of $[2 : n] \setminus J$
    \State $s \gets_{\text{LO}_j} \mathbb{Y}$ \Comment{Sample from $\text{LO}_j$}
    \State $J \gets J \cup \{j\}$
    \State $\boldsymbol{Z}[i]=s$
\EndIf
\EndFor
\State \Return $\boldsymbol{Z}$
\end{algorithmic}
\end{algorithm}

\section{Details in Proof of Theorem \ref{theorem_shuffle_IX}}\label{appendix_3}

\begin{fact}\label{fact1}
\(\Pr(X_1=x\mid C^{(s)}) \propto P(x)\,T_x^{(s)}\).
\end{fact}

\begin{proof}
Recall that \(X_1\sim P\), \(Y_1\mid X_1=x \sim \mathcal R_x\), and \(M^{(s)}\sim \mathrm{Mult}(s;Q)\) is independent of \((X_1,Y_1)\).  
We observe the histogram of size \(s\!+\!1\):
\[
C^{(s)} = e_{Y_1} + M^{(s)}, \quad
\text{i.e., } C^{(s)}_y = \mathbbm 1\{Y_1=y\} + M^{(s)}_y, \ \ \sum_y C^{(s)}_y = s+1.
\]
Fix any histogram \(c=(c_y)_{y\in\mathbb{Y}}\) with total count \(s+1\).  
By Bayes’ rule,
\[
\Pr(X_1=x\mid C^{(s)}=c)\ \propto\ P(x)\,\Pr(C^{(s)}=c\mid X_1=x).
\]

We compute the likelihood by conditioning on \(Y_1\):
\[
\Pr(C^{(s)}=c\mid X_1=x)
=\sum_{y\in\mathbb{Y}} \Pr(Y_1=y\mid X_1=x)\,\Pr\big(M^{(s)}=c-e_y\big)
=\sum_y \mathcal R_x(y)\,\Pr\big(M^{(s)}=c-e_y\big),
\]
where the summand is zero when \(c_y=0\).  
The multinomial probability mass function is
\[
\Pr(M^{(s)}=m)
=\frac{s!}{\prod_z m_z!}\prod_z Q(z)^{m_z}, \qquad \sum_z m_z = s.
\]
For each \(y\) with \(c_y\ge 1\),
\[
\Pr(M^{(s)}=c-e_y)
=\frac{s!}{\prod_z (c_z-\mathbbm 1\{z=y\})!}\prod_z Q(z)^{\,c_z-\mathbbm 1\{z=y\}}.
\]
Using \((c_y-1)! = c_y!/c_y\) and \(Q(z)^{c_z-\mathbbm 1\{z=y\}} = Q(z)^{c_z} Q(y)^{-1}\), we can separate the factor independent of \(y\):
\[
\Pr(M^{(s)}=c-e_y)
=\Bigg[\frac{s!}{\prod_z c_z!}\prod_z Q(z)^{c_z}\Bigg]\frac{c_y}{Q(y)}.
\]
Thus,
\[
\Pr(C^{(s)}=c\mid X_1=x)
=\Bigg[\frac{s!}{\prod_z c_z!}\prod_z Q(z)^{c_z}\Bigg]
\sum_y \mathcal R_x(y)\,\frac{c_y}{Q(y)}.
\]
The bracketed term does not depend on \(x\). Hence,
\[
\Pr(X_1=x\mid C^{(s)}=c)
\ \propto\ P(x)\sum_y c_y\,\frac{\mathcal R_x(y)}{Q(y)}
= P(x)\,T_x^{(s)},
\]
where we define
\[
T_x^{(s)} := \sum_y c_y\,w_x(y), \quad w_x(y):=\frac{\mathcal R_x(y)}{Q(y)}.
\]
Finally, replacing \(c\) by the random histogram \(C^{(s)}\) yields
\[
\Pr(X_1=x\mid C^{(s)}) \propto P(x)\,T_x^{(s)}, \quad
T_x^{(s)} = \sum_y C^{(s)}_y\,w_x(y).\qedhere
\]
\end{proof}

\begin{lemma}\label{lemma3}
Let $P$ be a probability mass function on a finite alphabet $\mathbb{X}$ and let
\[
\pi(x)\;=\;\frac{P(x)\bigl(1+\varepsilon_x\bigr)}{1+\bar\varepsilon},\qquad
\bar\varepsilon:=\sum_{u\in\mathbb{X}}P(u)\,\varepsilon_u,
\]
with $\sum_x P(x)\varepsilon_x$ finite and $1+\bar\varepsilon>0$, $1+\varepsilon_x>0$ for all $x$.
Then, for $\|\varepsilon\|_\infty$ sufficiently small, we have the second–order expansion
\begin{equation}\label{eq:quad}
D\!\bigl(\pi\Vert P\bigr)
=\frac{1}{2}\sum_{x\in\mathbb{X}} P(x)\,\bigl(\varepsilon_x-\bar\varepsilon\bigr)^2
+ R(\varepsilon),
\end{equation}
with a remainder satisfying
\begin{equation}\label{eq:remainder}
\bigl|R(\varepsilon)\bigr|
\;\le\; K\,\|\varepsilon\|_3^3,
\qquad
\|\varepsilon\|_3:=\Bigl(\sum_{x\in\mathbb{X}} P(x)\,|\varepsilon_x|^3\Bigr)^{\!\!1/3},
\end{equation}
where $K$ is a constant.
\end{lemma}

\begin{proof}
By definition,
\[
D(\pi\Vert P)=\sum_{x} \pi(x)\log\frac{\pi(x)}{P(x)}
=\sum_x \frac{P(x)(1+\varepsilon_x)}{1+\bar\varepsilon}\,
\Bigl(\log(1+\varepsilon_x)-\log(1+\bar\varepsilon)\Bigr).
\]
Write $a_x:=\varepsilon_x$ and $b:=\bar\varepsilon$ for brevity. Using the Taylor expansion
\[
\log(1+u)=u-\frac{u^2}{2}+r_3(u),\qquad |r_3(u)|\le C\,|u|^3
\quad\text{for }|u|\le u_0,
\]
we obtain, for all $x$ with $|a_x|,|b|\le u_0$,
\begin{align*}
\log(1+a_x)-\log(1+b)
&=(a_x-b)-\frac{1}{2}(a_x^2-b^2)+\bigl(r_3(a_x)-r_3(b)\bigr).
\end{align*}
Also expand the prefactor:
\[
\frac{1+a_x}{1+b}=1+(a_x-b)+\rho_x,
\qquad
\rho_x:=-\frac{(a_x-b)b}{1+b},
\]
so that $|\rho_x|\le C'(|a_x|+|b|)\,|b|$ for $|b|\le u_0$.
Hence
\begin{align*}
D(\pi\Vert P)
&=\sum_x P(x)\Bigl(1+(a_x-b)+\rho_x\Bigr)\Bigl[(a_x-b)-\frac{1}{2}(a_x^2-b^2)+r_3(a_x)-r_3(b)\Bigr]\\
&=\sum_x P(x)\Bigl[(a_x-b)-\frac{1}{2}(a_x^2-b^2)\Bigr]
+\sum_x P(x)(a_x-b)\Bigl[(a_x-b)-\frac{1}{2}(a_x^2-b^2)\Bigr]\\
&\quad+\sum_x P(x)\rho_x\Bigl[(a_x-b)-\frac{1}{2}(a_x^2-b^2)\Bigr]
+\sum_x P(x)\Bigl(1+(a_x-b)+\rho_x\Bigr)\bigl(r_3(a_x)-r_3(b)\bigr).
\end{align*}
The \emph{linear} term vanishes:
\[
\sum_x P(x)(a_x-b)=\sum_x P(x)a_x-b=0.
\]
For the \emph{quadratic} contribution, note that
\[
-\frac{1}{2}\sum_x P(x)(a_x^2-b^2)=-\frac{1}{2}\sum_x P(x)a_x^2+\frac{1}{2}b^2,
\qquad
\sum_x P(x)(a_x-b)^2
=\sum_x P(x)a_x^2-b^2.
\]
Therefore,
\[
-\frac{1}{2}\sum_x P(x)(a_x^2-b^2)+\sum_x P(x)(a_x-b)^2
=\frac{1}{2}\sum_x P(x)(a_x-b)^2.
\]
All the remaining pieces are cubic (or higher) in $\{a_x\}$ and $b$:
\begin{itemize}
\item $\sum_x P(x)(a_x-b)\cdot\bigl(-(a_x^2-b^2)/2\bigr)$ is $O\bigl(\sum P|a|^3+|b|^3\bigr)$;
\item $\sum_x P(x)\rho_x\bigl[(a_x-b)-\frac{1}{2}(a_x^2-b^2)\bigr]$ is bounded by a constant times
$\sum_x P(x)\bigl(|a_x|+|b|\bigr)|b|\bigl(|a_x-b|+|a_x|^2+|b|^2\bigr)$, hence $O\bigl(\sum P|a|^3+|b|^3\bigr)$;
\item $\sum_x P(x)\bigl(1+(a_x-b)+\rho_x\bigr)\bigl(r_3(a_x)-r_3(b)\bigr)$ is $O\bigl(\sum P|a_x|^3+|b|^3\bigr)$ by the uniform cubic bound on $r_3$.
\end{itemize}
Since \(b=\sum_x P(x)\,a_x\), the triangle inequality gives
\(
|b| \;\le\; \sum_x P(x)\,|a_x|.
\)
Applying Jensen’s inequality to the convex function \(\varphi(t)=|t|^3\) yields
\[
|b|^{3}
\;\le\;
\Bigl(\sum_x P(x)\,|a_x|\Bigr)^{\!3}
\;\le\;
\sum_x P(x)\,|a_x|^{3}.
\]
and thus each cubic contribution is bounded by a constant multiple of $\sum_x P(x)|a_x|^3=\|\varepsilon\|_3^3$.
Collecting terms yields \eqref{eq:quad}--\eqref{eq:remainder}.
\end{proof}

\begin{fact}\label{fact2}
Let $M^{(s)}=(M^{(s)}_y)_{y\in\mathbb{Y}}\sim\mathrm{Mult}(s;Q)$ with $Q(y)>0$ and
\[
c:=M^{(s)}-sQ,\qquad c_y:=M^{(s)}_y-sQ(y).
\]
Then
\[
\mathbb E[c]=0,\qquad 
\mathbb E\!\big[cc^\top\big]=\mathrm{Cov}(M^{(s)})=s\big(\mathrm{Diag}(Q)-QQ^\top\big).
\]   
\end{fact}

\begin{proof}
Write the multinomial vector as a sum of i.i.d.\ one-hot indicators.
Let $Y_1,\dots,Y_s$ be i.i.d.\ on $\mathbb{Y}$ with $\Pr(Y_t=y)=Q(y)$ and define
\[
\mathbf 1_t(y):=\mathbbm 1\{Y_t=y\},\qquad 
M^{(s)}_y=\sum_{t=1}^{s}\mathbf 1_t(y).
\]
Then
\[
\mathbb E[M^{(s)}_y]
=\sum_{t=1}^{s}\mathbb E[\mathbf 1_t(y)]
=\sum_{t=1}^{s}Q(y)=sQ(y),
\]
so $\mathbb E[c_y]=\mathbb E[M^{(s)}_y]-sQ(y)=0$, i.e.\ $\mathbb E[c]=0$.

For the second moments, for $a,b\in\mathbb{Y}$,
\[
\mathrm{Cov}(M^{(s)}_a,M^{(s)}_b)
=\sum_{t=1}^{s}\mathrm{Cov}\big(\mathbf 1_t(a),\mathbf 1_t(b)\big)
\]
because different trials $t\neq t'$ are independent. Now
\[
\mathrm{Cov}\big(\mathbf 1_t(a),\mathbf 1_t(b)\big)
=\begin{cases}
\mathrm{Var}(\mathbf 1_t(a))=Q(a)\big(1-Q(a)\big), & a=b,\\[4pt]
\mathbb E[\mathbf 1_t(a)\mathbf 1_t(b)]-\mathbb E[\mathbf 1_t(a)]\mathbb E[\mathbf 1_t(b)]
=0-Q(a)Q(b)=-Q(a)Q(b), & a\neq b,
\end{cases}
\]
since $\mathbf 1_t(a)\mathbf 1_t(b)=0$ a.s.\ for $a\neq b$ (mutually exclusive one-hot).
Therefore
\[
\mathrm{Cov}(M^{(s)}_a,M^{(s)}_b)
=s\Big(Q(a)\mathbf 1\{a=b\}-Q(a)Q(b)\Big).
\]
Stacking these entries gives the covariance matrix
\[
\mathrm{Cov}(M^{(s)})=s\big(\mathrm{Diag}(Q)-QQ^\top\big).
\]

Finally, because $c=M^{(s)}-\mathbb E[M^{(s)}]$ is the centered vector, we have
\[
\mathbb E[cc^\top]=\mathrm{Cov}(M^{(s)}),
\]
hence
\[
\mathbb E[cc^\top]=s\big(\mathrm{Diag}(Q)-QQ^\top\big).
\qedhere
\]    
\end{proof}

\begin{fact}[Variance of a Linear Statistic via Covariance Matrix]\label{fact3}
Let $\mathbb{Y}=\{y_1,\dots,y_m\}$ and define, for each $x$,
\[
w_x := \big(w_x(y_1),\dots,w_x(y_m)\big)^\top \in \mathbb R^m,
\qquad 
\bar w := \mathbb E_{x\sim P}[w_x].
\]
Let the covariance matrix
\[
\Sigma_w := \mathrm{Cov}_{x\sim P}(w_x)
= \mathbb E_x\!\big[(w_x-\bar w)(w_x-\bar w)^\top\big].
\]
For any (possibly random) vector $c=(c_{y_1},\dots,c_{y_m})^\top\in\mathbb R^m$ 
that is independent of $x$, define the linear statistic
\[
A_x := \sum_{y} c_y\, w_x(y) = c^\top w_x .
\]

\begin{enumerate}
\item[(1)] \textbf{(Conditional variance identity)}  
For any fixed $c$,
\[
\mathrm{Var}_{x\sim P}(A_x) \;=\; c^\top \Sigma_w\, c.
\]

\item[(2)] \textbf{(Taking expectation over $c$)}  
If $c$ is random and independent of $x$, then
\[
\mathbb E\!\left[\mathrm{Var}_{x\sim P}(A_x)\right]
= \mathbb E\big[c^\top \Sigma_w\, c\big]
= \mathrm{Tr}\!\big(\Sigma_w\,\mathbb E[c\,c^\top]\big).
\]
\end{enumerate}
\end{fact}

\begin{proof}
Since $\mathbb E_x[A_x]=c^\top \bar w$, we have
\[
A_x - \mathbb E_x[A_x] = c^\top (w_x-\bar w).
\]
Thus,
\[
\mathrm{Var}_{x}(A_x)
= \mathbb E_x\!\left[(c^\top (w_x-\bar w))^2\right]
= c^\top\,\mathbb E_x[(w_x-\bar w)(w_x-\bar w)^\top]\,c
= c^\top \Sigma_w c.
\]
If $c$ is random but independent of~$x$, taking expectation yields
\[
\mathbb E[\mathrm{Var}_x(A_x)]
= \mathbb E[c^\top \Sigma_w c]
= \mathrm{Tr}\big(\Sigma_w\,\mathbb E[c c^\top]\big),
\]
using $\mathrm{Tr}(AB)=\mathrm{Tr}(BA)$ and linearity of trace.
\end{proof}

\begin{lemma}\label{lemma_main}
$$\mathbb E\!\big[\mathrm{Var}_{x\sim P}(A_x^{(s)})\big]
=s\big(\overline{\chi^2}-\chi^2(\bar P\Vert Q)\big)\;+\;O(1).$$
\end{lemma}
\begin{proof}
Let $\boldsymbol{c}\in\mathbb R^{|\mathbb{Y}|}$ be defined by $c_y:=M^{(s)}_y-sQ(y)$.
Then $\mathbb E[\boldsymbol{c}]=\boldsymbol{0}$ and $\mathbb E[\boldsymbol{c}\boldsymbol{c}^\top]=s(\mathrm{Diag}(Q)-\boldsymbol{Q}\boldsymbol{Q}^\top)$, where $\mathrm{Diag}(Q)$ is a $|\bar{\mathbb{Y}}|\times |\bar{\mathbb{Y}}|$ matrix with diagnosis $\mathrm{Diag}(Q)_{ii}=Q(i)$, and $\boldsymbol{Q}$ is a column vector with $\boldsymbol{Q}_i=Q(i)$ (Fact \ref{fact2} in Appendix~\ref{appendix_3}).

Define the covariance matrix (with respect to $x\sim P$)
\[
\Sigma_w:=\mathrm{Cov}_{x\sim P}\big(w_x(\cdot)\big),\qquad
(\Sigma_w)_{ab}=\mathrm{Cov}_x\big(w_x(a),w_x(b)\big).
\]
For fixed $(M^{(s)},Y_1)$,
\[
\mathrm{Var}_{x\sim P}\!\big(A_x^{(s)}\big)
=\mathrm{Var}_{x\sim P}\!\Big(\sum_y c_y w_x(y)\Big)\;+\;\mathrm{Var}_{x\sim P}\!\big(w_x(Y_1)\big)
\;+\;2\,\mathrm{Cov}_{x\sim P}\!\Big(\sum_y c_y w_x(y),\,w_x(Y_1)\Big).
\]
Taking expectation in $(M^{(s)},Y_1)$, the cross term vanishes since $\mathbb E[c_y]=0$ and
$\boldsymbol{c}$ is independent of $(x,Y_1)$.
Therefore
\[
\mathbb E\big[\mathrm{Var}_{x\sim P}(A_x^{(s)})\big]
=\mathbb E\!\left[\mathrm{Var}_{x\sim P}\!\Big(\sum_y c_y w_x(y)\Big)\right]
+\mathbb E\!\left[\mathrm{Var}_{x\sim P}\big(w_x(Y_1)\big)\right].
\]
The first term equals (Fact \ref{fact3} in Appendix~\ref{appendix_3})
\[
\mathbb E\!\big[\boldsymbol{c}^\top \Sigma_w \boldsymbol{c}\big]
=\mathrm{Tr}\!\big(\Sigma_w\,\mathbb E[\boldsymbol{c}\boldsymbol{c}^\top]\big)
=s\,\mathrm{Tr}\!\big(\Sigma_w(\mathrm{Diag}(Q)-\boldsymbol{Q}\boldsymbol{Q}^\top)\big)
=s\,\mathrm{Tr}\!\big(\Sigma_w\cdot\mathrm{Diag}(Q)\big)
=s\sum_y Q(y)\,\mathrm{Var}_{x\sim P}(w_x(y)),
\]
because $\mathrm{Tr}(\boldsymbol{Q}\boldsymbol{Q}^\top\Sigma_w)=\mathrm{Var}_{x\sim P}(\sum_y Q(y)w_x(y))=0$, as
$\sum_y Q(y)w_x(y)=\sum_y \mathcal R_x(y)=1$ for all $x$.
Now observe that
\begin{align*}
\sum_{y} Q(y)\,\mathrm{Var}_x(w_x(y))
&=\sum_yQ(y)\mathbb{E}_x[w_x^2(y)]-\sum_yQ(y)\left( \mathbb{E}_x[w_x(y)] \right)^2  \\
&=\mathbb E_x\!\left[\sum_y \frac{\mathcal R_x(y)^2}{Q(y)}\right]
-\sum_y \frac{\bar P(y)^2}{Q(y)}\\
&=\overline{\chi^2}-\chi^2(\bar P\Vert Q).
\end{align*}

The second term $\mathbb E[\mathrm{Var}_x(w_x(Y_1))]$ is $O(1)$: the marginal of $Y_1$ is $\bar P$, and on a finite alphabet with $Q(y)>0$, the family $\{w_x(y)\}$ is uniformly bounded.  
Thus
\begin{equation}\label{eq:VarA-final}
\mathbb E\!\big[\mathrm{Var}_{x\sim P}(A_x^{(s)})\big]
=s\big(\overline{\chi^2}-\chi^2(\bar P\Vert Q)\big)\;+\;O(1).
\end{equation}
\end{proof}

}

\end{document}